\documentclass[a4paper]{article}
\usepackage{amsmath,amssymb,amsfonts,amsthm} % 使用AMS数学公式符号字体定理,页眉,扩展的color
\usepackage{graphicx} % Include this line if your
\usepackage{wrapfig} \usepackage{url} \usepackage[usenames]{color}
\def\b1{\mbox{\boldmath $1$}}

\theoremstyle{plain} \newtheorem{thm}{\bf Theorem}[section]
 \newtheorem{lem}[thm]{\bf Lemma}
 \newtheorem{defn}[thm]{\bf
  Definition} \theoremstyle{remark}

\makeatletter \@addtoreset{equation}{section} \makeatother \makeatletter
\@addtoreset{thm}{section} \makeatother
\allowdisplaybreaks %eqnarray公式太长，一页显示不完，公式就完全跑到下一页去了，上面一页好大一片空白，好难看，怎么样能让其在两页上显示
\usepackage[numbers,sort&compress]{natbib} % 使参考文献连续编号的引用以压缩方式出现
\usepackage[bookmarksnumbered, bookmarksopen,
colorlinks,citecolor=blue,linkcolor=blue]{hyperref}
\begin{document}
\date{}

\title{ Optimal Hybrid Dividend Strategy Under The Markovian Regime-Switching Economy}\author{XiaoXiao Zheng\thanks{School of Mathematical Sciences and LPMC, Nankai University}, Xin Zhang\thanks{School of Mathematical Sciences and LPMC, Nankai University, Tianjin 300071, P.R. China; Department of Mathematics, Southeast University, Nanjing, 210096, P.R. China; E-mail: nku.x.zhang@gmail.com} } \maketitle

\noindent{\bf Abstract}
In this paper, we consider the optimal dividend problem for a company. We describe the surplus process of the company by a diffusion model with regime switching. The aim of the company is to choose a dividend policy to maximize the expected total discounted payments until ruin. In this article, we consider a hybrid dividend strategy, that is, the company is allowed to conduct continuous dividend strategy as well as impulsive dividend strategy. In addition, we consider the change of economy, which is characterized by a markovian regime-switching, and under the setting of two regimes, we solve the problem and obtain the analytical solution for the value function.

\noindent
\\
\noindent {\it Keywords:} dividend strategy; impulse control; regime switching; quasi-variational inequalities

\newpage

\section{Introduction}
Over the past decades, optimal dividend problem has been a hot issue, and a large number of papers in this field have came out. The study of dividend problem has a realistic sense: for a joint-stock company, it has responsibility to pay dividends to its shareholder, therefore choosing a dividend strategy is important for this company.
The research of dividend problem stems from the work of \citet{de1957impostazione}. He is the first to suggest that a company should maximize the expected discounted dividend payout. In earlier research of this field, scholars focus on two kinds of dividend strategies. The first one is the constant barrier strategy.  In this model, we have a barrier b, splitting region into two parts. Under such policy, the surplus cannot cross the barrier $b$ at any time $t>0$, and when it hit the barrier, it will either stays at the barrier b or decreases below the barrier. Barrier strategy for the compound Poisson risk model has been studied by many scholars, \citet{gerber2006optimal} has showed that barrier strategy is optimal for the compound Poisson risk model when the initial value is below the barrier. For more results concern this topic, we could refer to other references. \citet{albrecher2005distribution} studied a model that allows for dividend payments under a linear barrier strategy. In this paper, partial integro-differential equations for Gerber and Shiu's discounted penalty function and for the moment generating function of the discounted sum of dividend payments were derived. A surplus process in the presence of a nonlinear dividend barrier was investigated in \citet{albrecher2002risk}. \citet{li2004class} treated renewal risk models with a constant dividend barrier. The second one is threshold strategy, specifically, dividends can be paid out at certain rate if the surplus exceeds a threshold. In the paper \citet{gerber2006optimal}, the author showed that the threshold strategy is optimal when the dividend rate is bounded and the claim distribution is exponential. Threshold dividend strategy under the classical compound Poisson model can be found in \citet{lin2006compound}.

Recently， stochastic control theory has been introduced to solve optimal dividend problem. HJB equation, QVI, and singular control as tools were used to deal with dividend problem from different aspects. There are two dividend strategies of interest------continuous dividend strategy and impulsive dividend strategy. \citet{asmussen1997controlled} considered optimal continuous dividend for diffusion model. In the paper \citet{azcue2005optimal}, for compound Poisson model, the author not only studied the optimal continuous dividend strategy but also considered the reinsurance policy. Impulsive dividend and reinsurance strategies for diffusion model can be found in paper \citet{cadenillas2006classical}. The case of compound Poisson model was studied in \citet{wei2010classical}. In our paper we consider a hybrid strategy, that is, we combine the continuous dividend strategy with impulsive dividend strategy. In reality such model make sense, it means that this company is allowed not only for continuously paying dividends but also for paying block dividends from time to time. Further we suppose the continuous dividend rate is bounded, while in \citet{sotomayor2011classical} the author consider the cases of both bound and unbound.

On the other hand, levy process, especially spectrally negative Levy process, is a powerful tool to characterize the dividend process, and there are many papers investigated this problem. In \citet{loeffen2008optimality}, under the classical optimal dividend control problem, the author studied the case in which the risk process is modeled by a general spectrally negative Levy process. \citet{loeffen2009optimal} considered an optimal dividends problem with transaction costs, and they assumed the reserves are described by a spectrally negative Levy process too.

During the past several years, there have been a large number of literatures considered regime switching problem. In general, regime switching was used to characterize the change of economic condition. Some scholars also studied regime switching in optimal dividend problem. ( see \citet{sotomayor2011classical}, \citet{wei2010classical}, \citet{wei2010optimal}).

In our paper, we articulate the problem of maximizing the expected utility of an company. Because of considering a compound dividend strategy(continuous dividend strategy as well as impulsive dividend strategy), the utility is consist of two parts which correspond to the two strategies respectively. On the other hand, we adopt Markovian regime-switching diffusion model to describe the surplus of the company, which is expected to reflect the economic circle's affection to the company's surplus.

The rest of the paper is organized as follows. In section 2, we introduce the model and some assumptions. Further we give two properties about value function. In Section 3, we introduce quasi-variational inequalities(QVI) and give the verification theorem. Section 4 contain our main result. In this section, we work out the candidate for the value function, and then we verify that this function satisfy the conditions in the verification theorem. In section 5, we give the conclusion.

%%%%%%%%%%%%%%%%%%%%%%%%%%%%%%%%%%%%%%%%%%%%%%%%%%%%%%%%%%%%%%%%%%%%%%%%%%%%%%%%%%%%%%%%%%%%
\section{Preliminaries }
\quad \quad In this paper, we assume the uncertainty is modeled by a probability space $(\Omega,\mathcal{F}, \mathcal{P})$. At first we assume that $\big\{\epsilon(t)\big\}_{t\geq0}$ be a homogenous finite-state continuous-time Markov chain, and for every $t\geq0$: $\epsilon(t)\in \mathbb{J}$, where $\mathbb{J}=\big\{1,2,...,N\big\}$ and $N\geq2$. We also assume Markov chain $\epsilon$ has a strongly irreducible generator $Q=[\lambda_{ij}]_{N\times N}$, where $ -\lambda_i:=\lambda_{ii}<0$ and $\sum_{j\in\mathbb{J}}\lambda_{ij}=0$ for every $i\in\mathbb{J}$. Here, Markov chain represent the change of the economic situation, it can be used to describe good economy and bad economy, or to describe the business circle: recovery, prosperity, recession and depression, or any other regimes of economy. Second, we consider a company with surplus process $\big\{X_t\big\}_{t>0}$ and the uncertainty of the surplus process is characterized by a Brownian motion $W$ and the markov chain defined above. We also assume that $\epsilon$ and $W$ are independent and we denote by $\mathbb{F}=\big\{\mathcal{F}_t\big\}_{t\geq0}$ the $\mathcal{P}$-augmentation of the filtration $\{\mathcal{F}_t^{(W,\epsilon)}\}_{t\geq0}$ generated by Brownian motion and markov chain, where $\mathcal{F}_t^{(W,\epsilon)}=\sigma\big\{W_s,\epsilon_s:0\leq s\leq t\big\}, t\geq0 $.

We suppose the company would conduct dividend strategy, we model the surplus of the company $X=\big\{X_t, t\geq0\big\}$ by stochastic differential equation (SDE) :
\begin{eqnarray}
\label{a1}
dX_t=\mu(\epsilon_t)dt+\sigma(\epsilon_t)dW_t-dZ_t-\sum_{n=1}^\infty I_{\{\tau_n<t\}}\xi_n
\end{eqnarray}
with initial value $X_0=x\geq0$ and initial state $\epsilon(0)=i$, where the adapted process $Z=\big\{Z_t, t\geq0\big\}$ and the sum of the sequence of the random variable $\big\{\xi_n\big\}_{n\in N}$ represent the cumulative amount of dividends paid out by the company up to time $t$.
\\
\\
Remark:

In contrast to some classical papers, our model consider a so-called hybrid dividend strategy, specifically, we assume that the company is allowed to carry out two sorts of dividend strategies: continuous dividend strategy and impulsive dividend strategy. In terms of realistic sense, hybrid strategy provides more ways of paying out dividends. In terms of methodology, such strategy is more complicated in treating than single strategy problem.
\\

In this paper, we assume $\big\{Z_t\big\}_{t\geq0}$ is absolutely continuous, and let $dZ_t=u_tdt$. Further, we suppose dividend rate $u_t$ is bounded, and we denote $L$ as its bound.

Besides we define the stopping time of bankruptcy
\[\Theta:=\inf\big\{t\geq0:X_t<0\big\}\]
and impose $X_t=0$ for $t\in [\Theta , \infty)$.

\begin{defn}
A triple
\[\pi:=\big(u^\pi,\Gamma^\pi,\xi^\pi\big)=\big(u^\pi(t) ; \tau_1^\pi\,\tau_2^\pi,... ; \xi_1^\pi,\xi_2^\pi,...\big)\]
is called an admissible control if and only if
\newline (i) $ u^\pi(t): [0,\infty)\mapsto[0,L] $ is an $\mathbb{F}$-adapted bounded process,
\newline (ii) $\tau_i^\pi, i=1,2,... $ are stopping times with respect to $\mathbb{F}$, and
\[0\leq\tau_1^\pi<\tau_2^\pi<\cdots<\tau_n^\pi<\cdots,  \quad a.s. \]
\newline (iii) The random variables $\xi_i^\pi, i=1,2,...$ are $\mathcal{F}_{\tau_i^\pi}$ measurable and $0\leq\xi_i^\pi\leq X_{\tau_i^\pi}$,
\newline (iv) For all $T\geq0$,
\[\mathcal{P}\big(\lim_{n\rightarrow \infty}\tau_n^\pi\leq T\big)=0.\]
\end{defn}
The set of all admissible controls is denoted by $\Pi$.\\

Suppose the utility function of the shareholder is given by $g(x)$, which belongs to the following class
\[\mathcal{G}:=\bigg\{g(x): \forall x\geq0,\quad g(0)<0,\quad g'(x)>0 \quad and \quad g''(x)<0\bigg\}.\]
Let $\Theta^\pi=\inf\big\{t\geq0:X^\pi_t<0\big\}$ be the ruin time under the policy $\pi$. We aim at choosing a hybrid dividend policy to maximize the expected total discounted dividend payments, then given initial surplus $x$ and initial state $i$, with each admissible $\pi=\big(u^\pi,\Gamma^\pi,\xi^\pi\big)$, we have cost function
\[V_\pi(x,i)=E_{x,i}\bigg\{\int_0^{\Theta^\pi} e^{-\delta t}u_t^\pi dt + \sum_{n=1}^\infty e^{-\delta\tau_n^\pi}g(\xi_n^\pi)I_{\{\tau_n^\pi<\Theta^\pi\}}\bigg\}.\]
Note that not only the continuous dividend but also the impulsive dividend are considered in our model, the cost function above should incorporate the two corresponding utilities. Then we can define the value function
\begin{eqnarray}
\label{1}
V(x,i)=\sup_{\pi\in\Pi}V_\pi(x,i).
\end{eqnarray}
The optimal control $\hat \pi=\big(u^{\hat \pi},\Gamma^{\hat\pi},\xi^{\hat\pi}\big)$ is a strategy under which we have the equality below:
\begin{eqnarray}
\label{2}
V(x,i)=V_{\hat \pi}(x,i).
\end{eqnarray}

For the need of the rest of the paper, we define
\[\mu^*:=\max_{i\in \mathbb{J}}\big\{\mu_i\big\} \quad and \quad \mu_*:=\min_{i\in \mathbb{J}}\big\{\mu_i\big\},\]
and a technical assumption(H) is made:
\[(H)\quad L<\mu_*  , \]
that is,  at each regime the bound of the dividend rate is smaller than the drift of the surplus.\\

Now we derive some properties of the value function.
\begin{lem}
Given $g\in\mathcal{G}$ is a utility function, then for each $i\in\mathbb{J}$, the value function satisfy :
\[V(x,i)\leq g'(0+)\big[x+\frac{\mu^*}{\delta}\big]+\frac{L}{\delta},\]
where $ \mu^*=\max_{i\in\mathbb{J}}\big\{\mu_i\big\}.$
\end{lem}
\begin{proof}
 For each $\pi \in \Pi$
 \[V_\pi(x,i)=E_{x,i}\bigg\{\int_0^{\Theta^\pi} e^{-\delta t}u_t^\pi dt + \sum_{n=1}^\infty e^{-\delta\tau_n^\pi}g(\xi_n^\pi)I_{\{\tau_n^\pi<\Theta^\pi\}}\bigg\}.\]
 First, we have
 \[E_{x,i}\bigg\{\int_0^{\Theta^\pi}e^{-\delta t}u_t^\pi dt\bigg\}\leq E_{x,i}\bigg\{\int_0^{\Theta^\pi}e^{-\delta t}Lds\bigg\}=\frac{L}{\delta}E_{x,i}\bigg\{1-e^{-\delta\Theta^\pi}\bigg\}\leq\frac{L}{\delta}.\]
 Next in virtue of $g(x)\leq g'(0+)x$, and for each $\pi\in\Pi$, let $D_t$ be the accumulated dividends process, and let $dD_t=0$ for $t\geq\Theta^\pi$. Thus we get
\begin{eqnarray*}
E_{x,i}\{\sum_{n=1}^\infty e^{-\delta\tau_n^\pi}g(\xi_n^\pi)I_{\{\tau_n^\pi<\Theta^\pi\}}\}
&\leq& g'(0+)E_{x,i}\bigg\{\sum_{n=1}^\infty e^{-\delta\tau_n^\pi}\xi_n^\pi I_{\{\tau_n^\pi<\Theta^\pi\}}\bigg\}\\
&=&g'(0+)E_{x,i}\bigg\{\int_0^\infty e^{-\delta s}d D_s\bigg\}\\
&=&g'(0+)E_{x,i}\bigg\{\int_0^\infty \delta e^{-\delta s}D_sds\bigg\}\\
&\leq& g'(0+)E_{x,i}\bigg\{\int_0^\infty \delta e^{-\delta s}(x+\mu^*s)ds\bigg\}\\
&=&g'(0+)\big(x+\frac{\mu^*}{\delta}\big),\\
\end{eqnarray*}
where $\mu^*=\max_{i\in\mathbb{J}}\big\{\mu_i\big\}$ . Then we have
\[V(x,i)\leq g'(0+)\big[x+\frac{\mu^*}{\delta}\big]+\frac{L}{\delta}.\]

\end{proof}

\begin{lem}
Given $g\in\mathcal{G}$ is a utility function, then for each $i\in\mathbb{J}$ and $0\leq x_1<x_2$, the value function satisfy :
\[V(x_2,i)-V(x_1,i)\geq g\big(x_2-x_1\big).\]
\end{lem}
\begin{proof}
 This proof is similar to the one in \citet{wei2010classical}, we omit it.
\end{proof}

\section{Quasi-Variational Inequalities }
Inspired by the work of \citet{cadenillas2006classical}, we introduce the quasi-variational inequalities(QVI). First for every continuous function $\phi : (0,\infty)\rightarrow \mathbb{R} $ and given state $i\in\mathbb{J}$,  we define the maximum operator $\mathcal{M}$ by
\begin{eqnarray}
\label{3}
\mathcal{M}\phi(x,i):=\sup\bigg\{\phi(x-u,i)+g(u):u\in \mathbb{R}, 0<u<x\bigg\}.
\end{eqnarray}
$\mathcal{M} V(x,i)$ represents the value of the policy that consists of choosing the best immediate intervention. Obviously we have
\[V(x,i)\geq\mathcal{M} V(x,i),\]
and if
\[V(x,i)>\mathcal{M} V(x,i),\]
this means $x$ is the position where it is not optimal to choose to pay block dividend, while if
\[V(x,i)=\mathcal{M} V(x,i),\]
this implies $x$ is the position where it is optimal to pay block dividend. Now let us introduce operator ${\mathcal{L}}$, given $i\in\mathbb{J}$, define:
\[\mathcal{L}_i(u)\phi(x,i)=\frac{1}{2}\sigma^2(i)\phi''(x,i)+\big[\mu(i)-u\big]\phi'(x,i)-\delta\phi(x,i)+\sum_{j=1}^N\lambda_{i,j}\phi(x,j). \]
\\

\begin{defn}
A function $\nu:[0,\infty)\mapsto[0,\infty)$ satisfies the quasi-variational inequalities of the control problem if for every $x\in[0,\infty)$, $i\in \mathbb{J}$ and $u\in[0,L]$
\begin{eqnarray}
\label{4}
\mathcal{L}_i(u)\nu(x,i)+u\leq0,\\
\label{5}
\nu(x,i)\geq\mathcal{M}\nu(x,i),\\
\label{6}
\big(\nu(x,i)-\mathcal{M}\nu(x,i)\big)\big(\max_{u\in [0,L]} \bigg\{\mathcal{L}_i(u)\nu(x,i)+u\bigg \}\big)=0.
\end{eqnarray}
\end{defn}
It is easy to observe that the solution of the QVI splits the interval $(0,\infty)$ into two disjoint regions: \\
(i) Continuation region
\[\mathcal{C}:=\bigg\{x\in(0,\infty):\nu(x)>\mathcal{M}\nu(x) \quad and \quad \max_{u\in [0,L]} \big\{\mathcal{L}_i(u)\nu(x,i)+u\big\}=0\bigg\}.\]
(ii) Intervention region
\[\Sigma:=\bigg\{x\in(0,\infty):\nu(x)=\mathcal{M}\nu(x) \quad and \quad \max_{u\in [0,L]} \big\{\mathcal{L}_i(u)\nu(x,i)+u\big\}\leq0\bigg\}.\]
It is obvious that continuation region is an open set and intervention region is a closed set. Given a solution $\nu$ to the QVI, we define the following strategy associated with this solution.
\\
\begin{defn}
The control $\pi^{\nu}=\big(u^\nu,\Gamma^\nu,\xi^\nu\big)=\big(u^\nu(t) ; \tau_1^\nu\,\tau_2^\nu,... ; \xi_1^\nu,\xi_2^\nu,...\big)$ is called the QVI control associated with $\nu$ if for each $i\in\mathbb{J}$ the associated state process $X^\nu$ given by \eqref{a1} satisfies
\begin{eqnarray}
\label{7}
&&\mathcal{P}\bigg\{u^\nu(t)\neq arg\big(\max_{u\in [0,L]}\big \{\mathcal{L}_i(u)\nu(X_t^\nu,i)+u \big\}\big), \quad X_t^\nu\in\mathcal{C}\bigg\}=0\\
\label{8}
&&\tau_1^\nu:=\inf\bigg\{t\geq0 : \nu(X_t^\nu,i)=\mathcal{M}\nu(X_t^\nu,i)\bigg\},\\
\label{9}
&&\xi_1^\nu:=arg\sup_{\eta>0,\eta\leq X_{\tau_1^\nu}^\nu}\bigg\{\nu(X^\nu(\tau_1^\nu)-\eta,i)+g(\eta)\bigg\},
\end{eqnarray}
and for each $n\geq2$,
\begin{eqnarray}
\label{10}
&&\tau_n^\nu:=\inf\bigg\{t\geq \tau_{n-1} : \nu(X_t^\nu,i)=\mathcal{M}\nu(X_t^\nu,i)\bigg\},\\
\label{11}
&&\xi_n^\nu:=arg\sup_{\eta>0,\eta\leq X_{\tau_n^\nu}^\nu}\bigg\{\nu(X^\nu(\tau_n^\nu)-\eta,i)+g(\eta)\bigg\},
\end{eqnarray}
with $\tau_0^\nu:=0$ and $\xi_0^\nu:=0$.
\end{defn}
\begin{thm}
Let $\nu(\cdot,i)\in C^1\big([0,\infty))$, $i\in\mathbb{J} $, and $\nu(\cdot,i)\in C^2\big([0,\infty)/\mathcal{N}_i\big)$, $i\in\mathbb{J} $, where $\mathcal{N}_i$ is finite subset of $(0,\infty)$, satisfies QVI \eqref{4}-\eqref{6} with $\nu(0,i)=0, i\in\mathbb{J}$. Assume there exists constant $U_i$, $0<U_i<\infty$, $i\in\mathbb{J} $ such that $\nu(x,i)$ is linear on $[U_i,\infty)$, then for every $x\in(0,\infty)$ and each $i\in\mathbb{J} $
\[V(x,i)\leq\nu(x,i).\]
Further, if the QVI control $\pi^\nu=\big(u^\nu,\Gamma^\nu,\xi^\nu\big)$ associated with $\nu$ is admissible, then $\nu$ coincides with the value function and the QVI control associated with $\nu$ is the optimal policy , $i.e.$,
\[V(x,i)=\nu(x,i)=V_{\pi^\nu}(x,i).\]
\end{thm}
\begin{proof}
Let $\pi=\big(u^\pi,\Gamma^\pi,\xi^\pi\big)$ is an admissible control, and $\tau_0^\pi:=0$, $ \xi_0^\pi:=0$ then for every $t\in [0, \infty )$
\begin{eqnarray*}
&&e^{-\delta(t\wedge\Theta^\pi\wedge\tau_n^\pi)}\nu(X_{(t\wedge\Theta^\pi\wedge\tau_n^\pi)},\epsilon_{(t\wedge\Theta^\pi\wedge\tau_n^\pi)})-\nu(x,i)\\
&&=\sum_{i=1}^n\bigg\{e^{-\delta(t\wedge\Theta^\pi\wedge\tau_i^\pi)}\nu(X_{(t\wedge\Theta^\pi\wedge\tau_i^\pi-)},\epsilon_{(t\wedge\Theta^\pi\wedge\tau_i^\pi)})-e^{-\delta(t\wedge\Theta^\pi\wedge\tau_{i-1}^\pi)}\nu(X_{(t\wedge\Theta^\pi\wedge\tau_{i-1}^\pi)},\epsilon_{(t\wedge\Theta^\pi\wedge\tau_{i-1}^\pi)})\bigg\}\\
&&\quad +\sum_{i=1}^n I_{\{\tau_i^\pi\leq t\wedge\Theta^\pi\} } e^{-\delta\tau_i^\pi}\big[\nu(X_{\tau_i^\pi},\epsilon_{\tau_i^\pi})-\nu(X_{\tau_i^\pi-},\epsilon_{\tau_i^\pi})\big].\\
\end{eqnarray*}
Since $\nu(\cdot,i)\in C^1\big([0,\infty))$ is twice continuously differentiable on $(0,\infty)$ with a possible exception of the finite point set $\mathcal{N}$, by Ito formula (see \citet{karatzas1991brownian})
\begin{eqnarray*}
&&e^{-\delta(t\wedge\Theta^\pi\wedge\tau_n^\pi)}\nu(X_{(t\wedge\Theta^\pi\wedge\tau_n^\pi)},\epsilon_{(t\wedge\Theta^\pi\wedge\tau_n^\pi)})-\nu(x,i)\\
&=&\sum_{i=1}^n\bigg\{\int_{[t \wedge \Theta^\pi\wedge \tau_{i-1}^\pi, t\wedge \Theta^\pi\wedge \tau_i^\pi )} e^{-\delta s} \mathcal{L}_{\epsilon_s}(u) \nu(X_{s},\epsilon_s)ds +  \int_{[t \wedge \Theta^\pi\wedge \tau_{i-1}^\pi, t\wedge \Theta^\pi\wedge \tau_i^\pi )}  e^{-\delta s} \nu_x(X_{s},\epsilon_s) \sigma(\epsilon_s)dW_s \\
&&+  \sum_{j=1}^N \int_{[t \wedge \Theta^\pi\wedge \tau_{i-1}^\pi, t\wedge \Theta^\pi\wedge \tau_i^\pi )} e^{-\delta s} [\nu(X_{s},j)-\nu(X_{s},\epsilon_s)]d\bar N_s^j      \bigg \} \\&&+\sum_{i=1}^n I_{\{\tau_i^\pi\leq t\wedge\Theta^\pi\} } e^{-\delta\tau_i^\pi}\big[\nu(X_{\tau_i^\pi},\epsilon_{\tau_i^\pi})-\nu(X_{\tau_i^\pi-},\epsilon_{\tau_i^\pi})\big],\\
\end{eqnarray*}
where $ \bar N_s $ is compensate poisson process. In virtue of QVI we have
\begin{eqnarray*}
&&e^{-\delta(t\wedge\Theta^\pi\wedge\tau_n^\pi)}\nu(X_{(t\wedge\Theta^\pi\wedge\tau_n^\pi)},\epsilon_{(t\wedge\Theta^\pi\wedge\tau_n^\pi)})-\nu(x,i)\\
&&\leq \sum_{i=1}^n \bigg\{ \int_{[t \wedge \Theta^\pi\wedge \tau_{i-1}^\pi, t\wedge \Theta^\pi\wedge \tau_i^\pi )}  e^{-\delta s} \nu_x(X_{s},\epsilon_s) \sigma(\epsilon_s)dW_s -  \int_{[t \wedge \Theta^\pi\wedge \tau_{i-1}^\pi, t\wedge \Theta^\pi\wedge \tau_i^\pi )} e^{-\delta s}u_s ds \\
 &&\quad +\sum_{j=1}^N \int_{[t \wedge \Theta^\pi\wedge \tau_{i-1}^\pi, t\wedge \Theta^\pi\wedge \tau_i^\pi )} e^{-\delta s} [\nu(X_{s},j)-\nu(X_{s},\epsilon_s)]d\bar N_s^j   \bigg  \} \\ &&\quad - \sum_{i=1}^n I_{\{\tau_i^\pi\leq t\wedge\Theta^\pi\} } e^{-\delta\tau_i^\pi} g(\xi_i^\pi).
\end{eqnarray*}
This inequality becomes an equality for the QVI control associated with $\nu$. Taking expectation
\begin{eqnarray*}
&&E_{x,i}\bigg\{e^{-\delta(t\wedge\Theta^\pi\wedge\tau_n^\pi)}\nu(X_{(t\wedge\Theta^\pi\wedge\tau_n^\pi)},\epsilon_{(t\wedge\Theta^\pi\wedge\tau_n^\pi)})-\nu(x,i)\bigg\}\\
&&\leq E_{x,i} \bigg\{\sum_{i=1}^n \big\{ \int_{[t \wedge \Theta^\pi\wedge \tau_{i-1}^\pi, t\wedge \Theta^\pi\wedge \tau_i^\pi )}  e^{-\delta s} \nu_x(X_{s},\epsilon_s) \sigma(\epsilon_s)dW_s -  \int_{[t \wedge \Theta^\pi\wedge \tau_{i-1}^\pi, t\wedge \Theta^\pi\wedge \tau_i^\pi )} e^{-\delta s}u_s ds \\
 &&\quad +\sum_{j=1}^N \int_{[t \wedge \Theta^\pi\wedge \tau_{i-1}^\pi, t\wedge \Theta^\pi\wedge \tau_i^\pi )} e^{-\delta s} [\nu(X_{s},j)-\nu(X_{s},\epsilon_s)]d\bar N_s^j \big\}  \\
 &&\quad - \sum_{i=1}^n I_{\{\tau_i^\pi\leq t\wedge\Theta^\pi\} } e^{-\delta\tau_i^\pi} g(\xi_i^\pi) \bigg \}.
\end{eqnarray*}
Since $\nu(x,i)\in C^1\big([0,U_i)\big)$, $\nu(x,i)$ is bounded on the interval $[0,U_i)$. On the other hand $\nu(x,i)$ is linear on $[U_i, \infty)$, then combining with the condition $P\big(\lim_{n\rightarrow\infty}\tau_n \leq T\big)=0 $ and dominated convergence theorem, we have
\begin{eqnarray*}
&&\lim_{n\rightarrow\infty}\bigg\{E_{x,i}\big[e^{-\delta(t\wedge\Theta^\pi\wedge\tau_n^\pi)}\nu(X_{(t\wedge\Theta^\pi\wedge\tau_n^\pi)},\epsilon_{(t\wedge\Theta^\pi\wedge\tau_n^\pi)})\big]-\nu(x,i)\bigg\}\\
&&=E_{x,i}\big[e^{-\delta(t\wedge\theta^\pi)}\nu(X_{t\wedge\theta^\pi},\epsilon_{t\wedge\theta^\pi})\big]-\nu(x,i),
\end{eqnarray*}
and
\begin{eqnarray*}
E_{x,i}\bigg\{\sum_{j=1}^N \int_{[t \wedge \Theta^\pi\wedge \tau_{i-1}^\pi, t\wedge \Theta^\pi\wedge \tau_i^\pi )} e^{-\delta s} [\nu(X_{s},j)-\nu(X_{s},\epsilon_s)]d\bar N_s^j  \bigg \}=0.
\end{eqnarray*}
Since $\nu(x,i)\in C^1\big([0,U_i)\big)$ and it's linear on $[U_i,\infty)$, we know $\nu_x(x,i)$ is bounded on $[0,U_i]$ and $\nu_x(x,i)$ is constant in the interval $[U_i,\infty)$ , then we get
\[E_{x,i}\bigg\{\sum_{i=1}^n\int_{[t \wedge \Theta^\pi\wedge \tau_{i-1}^\pi, t\wedge \Theta^\pi\wedge \tau_i^\pi )} e^{-\delta s}\nu_x(X_{s},\epsilon_s)\sigma(\epsilon_s)dW_s\bigg\}=0.\]
Thus we have
\begin{eqnarray*}
&&E_{x,i}\bigg[e^{-\delta(t\wedge\Theta^\pi)}\nu(X_{(t\wedge\Theta^\pi)},\epsilon_{(t\wedge\Theta^\pi)})\bigg]-\nu(x,i)\\
&&\leq E_{x,i}\bigg\{\sum_{i=1}^\infty\big[-I_{\{\tau_i^\pi\leq t\wedge\Theta^\pi\} } e^{-\delta\tau_i^\pi} g(\xi_i^\pi)- \int_{[t \wedge \Theta^\pi\wedge \tau_{i-1}^\pi, t\wedge \Theta^\pi\wedge \tau_i^\pi )} e^{-\delta s}u_s ds\big]\bigg\}.
\end{eqnarray*}
Now we let $t\longrightarrow\infty$
\begin{eqnarray*}
&&E_{x,i}\big[e^{-\delta\Theta^\pi}\nu(X_{\Theta^\pi},\epsilon_{\Theta^\pi})\big]-\nu(x,i)\\
&&\leq E_{x,i}\bigg\{\sum_{i=1}^\infty[-I_{\{\tau_i^\pi\leq \Theta^\pi\} } e^{-\delta\tau_i^\pi} g(\xi_i^\pi)- \int_{[ \Theta^\pi\wedge \tau_{i-1}^\pi,  \Theta^\pi\wedge \tau_i^\pi )} e^{-\delta s}u_s ds]\bigg\}\\
&&= -E_{x,i}\bigg\{\sum_{i=1}^\infty\big[I_{\{\tau_i^\pi\leq \Theta^\pi\} } e^{-\delta\tau_i^\pi} g(\xi_i^\pi)+ \int_0^{\Theta^\pi} e^{-\delta s}u_s ds\big]\bigg\},
\end{eqnarray*}
therefore
\begin{eqnarray*}
\nu(x,i)\geq E_{x,i}\bigg\{\sum_{i=1}^\infty I_{\{\tau_i^\pi\leq \Theta^\pi\} } e^{-\delta\tau_i^\pi} g(\xi_i^\pi)+ \int_0^{\Theta^\pi} e^{-\delta s}u_s ds \bigg\},
\end{eqnarray*}
where $\nu(X_{\Theta^\pi},\epsilon_{\Theta^\pi})= \nu(0,\epsilon_{\Theta^\pi})=0 $. Thus we have
\begin{eqnarray*}
\nu(x,i)\geq \sup_{\pi\in\Pi} E_{x,i}\bigg\{\sum_{i=1}^\infty I_{\{\tau_i^\pi\leq \Theta^\pi\} } e^{-\delta\tau_i^\pi} g(\xi_i^\pi)+ \int_0^{\Theta^\pi} e^{-\delta s}u_s ds \bigg\}=V(x,i).
\end{eqnarray*}
And the inequality becomes an equality for the QVI control associated with $\nu$.
\end{proof}

\section{The Solution of The QVI and The Optimal Policy}
In this section we set out to solve the QVI \eqref{4}-\eqref{6}, and then verify that the solution satisfies the theorem 3.3. We let $g(x)=x-K$, $K$ represents the fixed cost. Note that different from \citet{cadenillas2006classical} and \citet{wei2010classical}, in which they consider the fixed cost $K$ as well as the proportional cost $k$, i.e. let  $g(x)=kx-K$. In this paper we consider the fixed cost only, in fact we will see in the following that taking the proportional cost into account is nothing but more tedious classification and calculation. Our essential method and verification can be fully illustrated by the simpler model.
\\

Now we want to find a function $\nu(x,i), i\in\mathbb{J}$ as a candidate for the value function and verify the candidate satisfies the condition of theorem 3.3. For such function $\nu(x,i)$, define
\begin{eqnarray}
\label{a2}
B_i:=\inf\bigg \{x\geq0 : \nu(x,i)=\mathcal{M}\nu(x,i),\quad i\in\mathbb{J}\bigg\}.
\end{eqnarray}
We conjecture $\nu(x,i)$, $i\in\mathbb{J}$ satisfied QVI to be continuously differentiable function, and on the interval $[0,B_i]$ function $\nu'(x,i)$, $i\in\mathbb{J}$ is convex. In addition, we assume $\nu(0,i)=0$, $i\in\mathbb{J}$. Note we will prove later these conjectures are satisfied .

Obviously the assumption of convexity of $\nu'(x,i), i\in\mathbb{J}$ implies that the equation
\begin{eqnarray}
\label{12}
\nu'(x,i)=1, \quad\quad\quad\quad i\in\mathbb{J},
\end{eqnarray}
has at most two roots on $[0,B_i]$. Indeed the proposition below tells us that E.q.\eqref{12} do exist only one root in the interval $[0,B_i)$.
\\
\begin{thm}
Let $B_i$, $i\in\mathbb{J}$ be defined by \eqref{a2}.  $\nu(x,i)$, $i\in\mathbb{J}$ is continuously differentiable function and $\nu'(x,i)$, $i\in\mathbb{J}$ is convex on $[0,B_i)$. Then there exists only one root $b_i\in[0,B_i)$, $i\in\mathbb{J}$ of the equation
\[\nu'(x,i)=1, \quad\quad\quad\quad  i\in\mathbb{J}.\]
Moreover if the solution to QVI is unique, then for $x\geq B_i$:
\[\nu(x,i)=\nu(b_i,i)+x-b_i+K .\]
\end{thm}
\begin{proof}
    The proof of this theorem is similar to Theorem 3.2 in \citet{cadenillas2006classical}, so we do not repeat it here.
\end{proof}

From the analysis above we know under assumption in theorem 4.1, the continuation region is given by $\mathcal{C}_i=[0, B_i)$, and when the surplus exceeds level $B_i$, dividend will be paid such that the surplus jumps to $b_i\in[0,B_i)$. Concretely, when $x\in \mathcal{C}_i=[0,B_i)$ we have
\begin{eqnarray}
\label{a3}
\max_{u\in[0,L]}\bigg\{\mathcal{L}_i(u)\nu(x,i)+u\bigg\}=0,
\end{eqnarray}
that is
\begin{eqnarray}
\label{a4}
\frac{1}{2}\sigma^2(i)\nu''(x,i)+\mu(i)\nu'(x,i)-\delta\nu(x,i)+\max_{u\in[0,L]}\bigg \{u\big(1-\nu'(x,i)\big)\bigg\}=-\sum_{j=1}^N\lambda_{i,j}\nu(x,j).
\end{eqnarray}
Thus, if $v'(X_t,\epsilon_t)\leq1$,
\begin{eqnarray}
\label{a5}
\frac{1}{2}\sigma^2(i)\nu''(x,i)+\mu(i)\nu'(x,i)-\delta\nu(x,i)+L\big(1-\nu'(x,i)\big)=-\sum_{j=1}^N\lambda_{i,j}\nu(x,j),
\end{eqnarray}
if $v'(X_t,\epsilon_t)\geq 1$,
\begin{eqnarray}
\label{a6}
\frac{1}{2}\sigma^2(i)\nu''(x,i)+\mu(i)\nu'(x,i)-\delta\nu(x,i)=-\sum_{j=1}^N\lambda_{i,j}\nu(x,j).
\end{eqnarray}
When  $x\in \Sigma_i=[B_i,\infty)$, we have
\begin{eqnarray}
\label{a7}
\nu(x,i)=\nu(b_i,i)+x-b_i+K.
\end{eqnarray}
Note in this  case $\nu'(x,i)=1$.
\\

After making a simple analysis we can see that under the assumption above, $\nu'(0,i)\geq 0$, $i\in\mathbb{J}$ established. Indeed we have $\nu'(B_i,i)=1$, $i\in\mathbb{J}$ and for some $b_i\in[0,B_i)$, $\nu'(b_i,i)=1$, $i\in\mathbb{J}$, then combining with the convexity of $\nu(x,i)$, $i\in\mathbb{J}$, we know $\nu'(0,i)<1$ is impossible. Therefore there are two cases have to be considered :\\
\\
\textbf{case} 1:  $\nu'(0,i)> 1$, $i\in\mathbb{J}$.

In this case, interval $[0,\infty)$ can be split into three sections : $[0,b_i)$, $[b_i,B_i)$, $[B_i,\infty)$.\\
(1) When $x\in[0,b_i)$, $\nu'(x,i)>1$. E.q.\eqref{a6} gives the following differential equations:
\begin{eqnarray}
\label{a8}
\frac{1}{2}\sigma^2(i)\nu''(x,i)+\mu(i)\nu'(x,i)-\delta\nu(x,i)=-\sum_{j=1}^N\lambda_{i,j}\nu(x,j).
\end{eqnarray}
(2)When $x\in[b_i,B_i)$, $\nu'(x,i)\leq1$. E.q.\eqref{a5} gives the following differential equations:
\begin{eqnarray}
\label{a9}
\frac{1}{2}\sigma^2(i)\nu''(x,i)+\mu(i)\nu'(x,i)-\delta\nu(x,i)+L\big(1-\nu'(x,i)\big)=-\sum_{j=1}^N\lambda_{i,j}\nu(x,j).
\end{eqnarray}
(3) When $x\in[B_i,\infty)$
\begin{eqnarray}
\label{a10}
\nu(x,i)=\nu(b_i,i)+x-b_i+K.
\end{eqnarray}
\\
\textbf{case 2}: $\nu'(0,i)=1$, $i\in\mathbb{J}$.

In this case, interval $[0,\infty)$ can be split into two sections : $[0,B_i)$, $[B_i,\infty)$.\\
(1)When $x\in[0,B_i)$, $\nu'(x,i)\leq1$ E.q.\eqref{a5} gives the following differential equations:
\begin{eqnarray}
\label{a11}
\frac{1}{2}\sigma^2(i)\nu''(x,i)+\mu(i)\nu'(x,i)-\delta\nu(x,i)+L\big(1-\nu'(x,i)\big)=-\sum_{j=1}^N\lambda_{i,j}\nu(x,j).
\end{eqnarray}
(2) When $x\in[B_i,\infty)$
\begin{eqnarray}
\label{a12}
\nu(x,i)=\nu(b_i,i)+x-b_i+K=x-K,
\end{eqnarray}
since $\nu(b_i,i)=\nu(0,i)=0$.\\

For simplicity, we suppose in the remainder of this section that the economy changes only between two regimes, that is $\mathbb{J}=\{1,2\}$. Combining with the analysis above, we have three possible cases: $\nu'(0,i)>1$ for both $i\in\mathbb{J}$; $\nu'(0,i_0)=1$ and $\nu'(0,3-i_0)=1$ for some $i_0\in\{1,2\}$; and $\nu'(0,i)=1 $ for both $i\in\mathbb{J}$.

In the following, we need a lemma , which is cited from \cite{sotomayor2011classical}.

\begin{lem}
For $i\in\mathbb{J}$, consider the real function $\phi_i(z)=-\frac{1}{2}\sigma_i^2z^2-\bar\mu_iz+(\lambda_i+\delta)$ where $\bar\mu_i$ is a function of $\mu_i$. Since $\sigma_1$, $\sigma_2$, $\lambda_1$ and $\lambda_2$ are positive, the equation $\phi_1(z)\phi_2(z)=\lambda_1\lambda_2$ has four real roots such that $z_1<z_2<0<z_3<z_4$.
\end{lem}

According to our conjecture about $\nu'(x,i)$, $i\in\{1,2\}$, there exist two thresholds $b_i<B_i$, $i\in\{1,2\}$ such that $\nu'(b_i,i)=\nu'(B_i,i)=1$, $i\in\{1,2\}$. We should note that the relationship of $b_i$ and $B_i$, $i\in\{1,2\}$, depends on the relations of the parameters in model. Here, we only consider two cases : $b_1\leq b_2<B_1\leq B_2$ and $b_1<B_1<b_2<B_2$, that is we discuss the optimization problem among models of which parameters satisfy certain relation such that we have $b_1\leq b_2<B_1\leq B_2$ and $b_1<B_1<b_2<B_2$. For the case of $b_2\leq b_1<B_2\leq B_1$, $b_2<B_2<b_1<B_1$, $b_1\leq b_2<B_2\leq B_1$ and $b_2\leq b_1<B_1\leq B_2$, they can be treated in a similar way.

\subsection{The case of $b_1\leq b_2<B_1\leq B_2$}
In this subsection we assume $b_1\leq b_2<B_1\leq B_2$. Considering the relationship between $\nu'(0,i)$, $i\in\{1,2\}$ and $1$, we have three cases, $\nu'(0,i)>1$ for both $i=1,2$; $\nu'(0,i_0)=1$ and $\nu'(0,3-i_0)>1$ for some $i_0\in\{1,2\}$; and $\nu'(0,i)=1 $ for both $i=1,2$.
\\

\textbf{Case 1}: $\nu'(0,1)>1$ and $\nu'(0,2)>1$ .\\

According to the discussion above, we need to consider five possibilities: $x\in[0,b_1)$; $x\in[b_1,b_2)$; $x\in[b_2,B_1)$; $x\in[B_1,B_2)$ and $x\in[B_2,\infty)$.\\

 When $x\in[0,b_1)$, we have $\nu'(x,1)>1$ and $\nu'(x,2)>1$, E.q.\eqref{a6} gives the following system of differential equations:
\begin{eqnarray}
  \left\{\begin{array}{l}
  \label{13}
\frac{1}{2}\sigma^2(1)\nu''(x,1)+\mu(1)\nu'(x,1)-\delta\nu(x,1)=\lambda_1\nu(x,1)-\lambda_1\nu(x,2),\\
\\
\frac{1}{2}\sigma^2(2)\nu''(x,2)+\mu(1)\nu'(x,2)-\delta\nu(x,2)=\lambda_2\nu(x,2)-\lambda_2\nu(x,1),
\end{array}\right.
\end{eqnarray}
where $\lambda_1=-\lambda_{11}$, $\lambda_2=-\lambda_{22}$. Consider the characteristic equation for \eqref{13}, $\phi_1^1(\dot\alpha)\phi_2^1(\dot\alpha)=\lambda_1\lambda_2$, where $\phi_i^1(\dot\alpha)=-\frac{1}{2}\sigma_i^2\dot\alpha^2-\mu_i\dot\alpha+(\lambda_i+\delta)$, $i\in\{1,2\}$. Lemma 4.2 proves that $\phi_1^1(\dot\alpha)\phi_2^1(\dot\alpha)=\lambda_1\lambda_2$ has four real roots: $\dot\alpha_1<\dot\alpha_2<0<\dot\alpha_3<\dot\alpha_4$. Then the solution of the E.q.\eqref{13} is
\begin{eqnarray}
  \left\{\begin{array}{l}
  \label{14}
\nu(x,1)=\dot A_1e^{\dot\alpha_1x}+\dot A_2e^{\dot\alpha_2x}+\dot A_3e^{\dot\alpha_3x}+\dot A_4e^{\dot\alpha_4x},\\
\\
\nu(x,2)=\dot B_1e^{\dot\alpha_1x}+\dot B_2e^{\dot\alpha_2x}+\dot B_3e^{\dot\alpha_3x}+\dot B_4e^{\dot\alpha_4x},
\end{array}\right.
\end{eqnarray}
where for each $j=1,2,3,4$,\\
\begin{eqnarray}
\label{z1} \dot B_j=\frac{\phi_1^1(\dot\alpha_j)}{\lambda_1}\dot A_j=\frac{\lambda_2}{\phi_2^1(\dot\alpha_j)}\dot A_j.
\end{eqnarray}
\\

When $x\in[b_1,b_2)$, we have $\nu'(x,1)\leq1$ and $\nu'(x,2)>1$, E.q.\eqref{a6},\eqref{a5} give the following system of differential equations:
\begin{eqnarray}
  \left\{\begin{array}{l}
  \label{15}
\frac{1}{2}\sigma^2(1)\nu''(x,1)+\mu(1)\nu'(x,1)-\delta\nu(x,1)+L\big(1-\nu'(x,1)\big)=\lambda_1\nu(x,1)-\lambda_1\nu(x,2),\\
\\
\frac{1}{2}\sigma^2(2)\nu''(x,2)+\mu(1)\nu'(x,2)-\delta\nu(x,2)=\lambda_2\nu(x,2)-\lambda_2\nu(x,1).
\end{array}\right.
\end{eqnarray}
Consider the characteristic equation for \eqref{15}, $\phi_1^2(\widetilde\alpha)\phi_2^2(\widetilde\alpha)=\lambda_1\lambda_2$, where $\phi_1^2(\widetilde\alpha)=-\frac{1}{2}\sigma_1^2\widetilde\alpha^2-(\mu_1-L)\widetilde\alpha+(\lambda_1+\delta)$, $\phi_2^2(\widetilde\alpha)=-\frac{1}{2}\sigma_2^2\widetilde\alpha^2-\mu_1\widetilde\alpha+(\lambda_2+\delta)$. Lemma 4.2 proves that $\phi_1^2(\alpha)\phi_2^2(\alpha)=\lambda_1\lambda_2$ has four real roots: $\widetilde\alpha_1<\widetilde\alpha_2<0<\widetilde\alpha_3<\widetilde\alpha_4$. Then the solution of the E.q.\eqref{15} is
\begin{eqnarray}
  \left\{\begin{array}{l}
  \label{16}
\nu(x,1)=\widetilde A_1e^{\widetilde\alpha_1(x-b_2)}+\widetilde A_2e^{\widetilde\alpha_2(x-b_2)}+\widetilde A_3e^{\widetilde\alpha_3(x-b_2)}+\widetilde A_4e^{\widetilde\alpha_4(x-b_2)}+F_1,\\
\\
\nu(x,2)=\widetilde B_1e^{\widetilde\alpha_1(x-b_2)}+\widetilde B_2e^{\widetilde\alpha_2(x-b_2)}+\widetilde B_3e^{\widetilde\alpha_3(x-b_2)}+\widetilde B_4e^{\widetilde\alpha_4(x-b_2)}+F_2,\\
\end{array}\right.
\end{eqnarray}
where $F_i:=\frac{(\lambda_2+(2-i)\delta)L}{(\lambda_i+\delta)(\lambda_2+\delta)-\lambda_1\lambda_2}$, $i=1,2$ and for $j=1,2,3,4,$
\begin{eqnarray}
\label{z2} \widetilde B_j=\frac{\phi_1^2(\widetilde\alpha_j)}{\lambda_1}\widetilde A_j=\frac{\lambda_2}{\phi_2^2(\widetilde\alpha_j)}\widetilde A_j.
\end{eqnarray}
\\

When $x\in[b_2,B_1)$, we have $\nu'(x,1)<1$ and $\nu'(x,2)\leq1$. E.q.\eqref{a5} gives the following system of differential equations:
\begin{eqnarray}
  \left\{\begin{array}{l}
  \label{17}
\frac{1}{2}\sigma^2(1)\nu''(x,1)+\mu(1)\nu'(x,1)-\delta\nu(x,1)+L\big(1-\nu'(x,1)\big)=\lambda_1\nu(x,1)-\lambda_1\nu(x,2),\\
\\
\frac{1}{2}\sigma^2(2)\nu''(x,2)+\mu(2)\nu'(x,2)-\delta\nu(x,2)+L\big(1-\nu'(x,2)\big)=\lambda_2\nu(x,2)-\lambda_2\nu(x,1).
\end{array}\right.
\end{eqnarray}
Consider the characteristic equation for \eqref{17}, $\phi_1^3(\hat\alpha)\phi_2^3(\hat\alpha)=\lambda_1\lambda_2$, where $\phi_i^3(\hat\alpha)=-\frac{1}{2}\sigma_i^2\hat\alpha^2-(\mu_i-L)\hat\alpha+(\lambda_i+\delta)$, $i=1,2$. Lemma 4.2 proves that $\phi_1^3(\alpha)\phi_2^3(\alpha)=\lambda_1\lambda_2$ has four real roots: $\hat\alpha_1<\hat\alpha_2<0<\hat\alpha_3<\hat\alpha_4$. Then the solution of the E.q.\eqref{17} is
\begin{eqnarray}
  \left\{\begin{array}{l}
  \label{18}
\nu(x,1)=\hat A_1e^{\hat\alpha_1(x-B_1)}+\hat A_2e^{\hat \alpha_2(x-B_1)}+\hat A_3e^{\hat\alpha_3(x-B_1)}+\hat A_4e^{\hat\alpha_4(x-B_1)}+\frac{L}{\delta},\\
\\
\nu(x,2)=\hat B_1e^{\hat\alpha_1(x-B_1)}+\hat B_2e^{\hat\alpha_2(x-B_1)}+\hat B_3e^{\hat\alpha_3(x-B_1)}+\hat B_4e^{\hat\alpha_4(x-B_1)}+\frac{L}{\delta},\\
\end{array}\right.
\end{eqnarray}
and for $j=1,2,3,4,$\\
\begin{eqnarray}
\label{z3} \hat B_j=\frac{\phi_1^3(\hat\alpha_j)}{\lambda_1}\hat A_j=\frac{\lambda_2}{\phi_2^3(\hat\alpha_j)}\hat A_j.
\end{eqnarray}
\\

When $x\in[B_1,B_2)$, we have $\nu'(x,1)=1$ and $\nu'(x,2)<1$. E.q.\eqref{a5},\eqref{a7} give the following system of differential equations:
\begin{eqnarray}
  \left\{\begin{array}{l}
  \label{19}
\nu(x,1)=\nu(b_1,1)+x-b_1-K,\\
\\
\frac{1}{2}\sigma^2(2)\nu''(x,2)+\mu(2)\nu'(x,2)-\delta\nu(x,2)+L\big(1-\nu'(x,2)\big)=\lambda_2\nu(x,2)-\lambda_2\nu(x,1),
\end{array}\right.
\end{eqnarray}
where $\nu(b_1,1)$ can be calculated by E.q.\eqref{16}. Consider the characteristic equation for the second equation above, $\phi^4(\breve \alpha)=\frac{1}{2}\sigma_2^2\breve \alpha^2+(\mu_2-L)\breve \alpha-(\lambda_2+\delta)=0$. It has two real roots: $\breve\alpha_1<\breve\alpha_2$, then the solution of the ODE is:
\[\nu(x,2)=\breve B_1e^{\breve \alpha_1(x-B_2)}+\breve B_2e^{\breve\alpha_2(x-B_2)}+U(x),\]
where $U(x)=\frac{\lambda_2}{\lambda_2+\delta}x+\frac{1}{\lambda_2+\delta}\bigg\{(\mu_2-L)\frac{\lambda_2}{\lambda_2+\delta}+L+\lambda_2\big[\nu(b_1,1)-b_1-K\big]\bigg\}$.
Then the E.q.\eqref{19} has solution:
\begin{eqnarray}
  \left\{\begin{array}{l}
  \label{20}
\nu(x,1)=\nu(b_1,1)+x-b_1-K,\\
\\
\nu(x,2)=\breve B_1e^{\breve \alpha_1(x-B_2)}+\breve B_2e^{\breve\alpha_2(x-B_2)}+U(x).
\end{array}\right.
\end{eqnarray}
\\

When $x\in[B_2,\infty)$. From the equation \eqref{a7}, we have
\begin{eqnarray}
  \left\{\begin{array}{l}
  \label{21}
\nu(x,1)=\nu(b_1,1)+x-b_1-K,\\
\\
\nu(x,2)=\nu(b_2,2)+x-b_2-K,
\end{array}\right.
\end{eqnarray}
where $\nu(b_1,1)$ and $\nu(b_2,2)$ can be calculated by E.q.\eqref{16} and E.q.\eqref{18} respectively.\\

In order to find the thresholds $b_1,b_2,B_1,B_2$, and the coefficients $\dot A_1,\dot A_2,\dot A_3,\dot A_4$, $\widetilde A_1,\widetilde A_2,\widetilde A_3,\widetilde A_4$, $\hat A_1,\hat A_2,\hat A_3,\hat A_4$, $\breve B_1,\breve B_2$. We suppose the smooth fit condition hold, and combined with $\nu'(b_i,i)=1$ for $i=1,2$, we have the equations following:
\begin{eqnarray}
\label{a13}
\nonumber &&\nu(0,1)=0, \quad\nu(0,2)=0, \quad\nu(b_1+,1)=\nu(b_1-,1),\quad  \nu(b_1+,2)=\nu(b_1-,2)\\
\nonumber &&\nu'(b_1+,1)=1, \quad\nu'(b_1-,1)=1, \quad\nu'(b_1+,2)=\nu'(b_1-,2), \quad\nu(b_2+,1)=\nu(b_2-,1)\\
\nonumber &&\nu(b_2+,2)=\nu(b_2-,2), \quad \nu'(b_2+,1)=\nu'(b_2-,1), \quad \nu'(b_2+,2)=1, \quad \nu'(b_2-,2)=1\\
\nonumber &&\nu(B_1+,1)=\nu(B_1-,1), \quad \nu(B_1+,2)=\nu(B_1-,2), \quad \nu'(B_1+,2)=\nu'(B_1-,2), \quad \nu'(B_1-,1)=1\\
&&\nu(B_2+,2)=\nu(B_2-,2), \quad \nu'(B_2-,2)=1.
\end{eqnarray}
Meanwhile the coefficients $\dot B_1,\dot B_2,\dot B_3,\dot B_4$, $\widetilde B_1,\widetilde B_2,\widetilde B_3,\widetilde B_4$, $\hat B_1,\hat B_2,\hat B_3,\hat B_4$ can be obtained from equation \eqref{z1},\eqref{z2},\eqref{z3}.

Note that we assume $\nu'(0,1)>1$ and $\nu'(0,2)>1$. This condition is satisfied if and only if the coefficients found through the system of E.q.\eqref{14} satisfy:
\begin{eqnarray}
  \left\{\begin{array}{l}
  \label{22}
\dot A_1\dot \alpha_1+\dot A_2\dot\alpha_2+\dot A_3\dot\alpha_3+\dot A_4\dot\alpha_4>1,
\\
\dot A_1\varphi_1^1(\dot\alpha_1)\alpha_1+\dot A_2\varphi_1^1(\dot\alpha_2)\alpha_2+\dot A_3\varphi_1^1(\dot\alpha_3)\alpha_3+\dot A_4\varphi_1^1(\dot\alpha_4)\alpha_4>\lambda_1.
\end{array}\right.
\end{eqnarray}
\\

\textbf{Case 2}: For some $i_0\in\{1,2\}$ $\nu'(0,i_0)=1$, $\nu'(0,3-i_0)>1$.\\

Without loss of generality we assume $i_0=1$, we have $\nu'(0,1)=1$, $\nu'(0,2)>1$. The case of $i_0=2$ has the similar treatment. Under this assumption,we have $b_1=0$ and we need to consider four possibilities: $x\in[0,b_2)$; $x\in[b_2,B_1)$; $x\in[B_1,B_2)$ and $x\in[B_2,\infty)$.\\

When $x\in[0,b_2)$, $\nu'(x,1)\leq1$ and $\nu'(x,2)>1$, we have the following system of differential equations:
\begin{eqnarray}
  \left\{\begin{array}{l}
  \label{23}
\frac{1}{2}\sigma^2(1)\nu''(x,1)+\mu(1)\nu'(x,1)-\delta\nu(x,1)+L\big(1-\nu'(x,1)\big)=\lambda_1\nu(x,1)-\lambda_1\nu(x,2),\\
\\
\frac{1}{2}\sigma^2(2)\nu''(x,2)+\mu(1)\nu'(x,2)-\delta\nu(x,2)=\lambda_2\nu(x,2)-\lambda_2\nu(x,1).
\end{array}\right.
\end{eqnarray}
Consider the characteristic equation for \eqref{23}, $\phi_1^2(\theta)\phi_2^2(\theta)=\lambda_1\lambda_2$, where $\phi_1^2$ and $\phi_2^2$ have been defined in case 1, and by the similar way we know the solution of the E.q.\eqref{23} is
\begin{eqnarray}
  \left\{\begin{array}{l}
  \label{24}
\nu(x,1)=C_1e^{\theta_1x}+C_2e^{\theta_2x}+C_3e^{\theta_3x}+C_4e^{\theta_4x}+F_1,\\
\\
\nu(x,2)=D_1e^{\theta_1x}+D_2e^{\theta_2x}+D_3e^{\theta_3x}+D_4e^{\theta_4x}+F_2,
\end{array}\right.
\end{eqnarray}
where $F_i:=\frac{(\lambda_2+(2-i)\delta)L}{(\lambda_i+\delta)(\lambda_2+\delta)-\lambda_1\lambda_2}$, $i=1,2$ and for $j=1,2,3,4,$
\begin{eqnarray}
\label{z4}
D_j=\frac{\phi_1^2(\theta_j)}{\lambda_1}C_j=\frac{\lambda_2}{\phi_2^2(\theta_j)}C_j,
\end{eqnarray}
where $\theta_1<\theta_2<0<\theta_3<\theta_4$ are characteristic roots.
\\

When $x\in[b_2,B_1)$, $\nu'(x,1)<1$ and $\nu'(x,2)\leq1$, we have the following system of differential equations:
\begin{eqnarray}
  \left\{\begin{array}{l}
  \label{25}
\frac{1}{2}\sigma^2(1)\nu''(x,1)+\mu(1)\nu'(x,1)-\delta\nu(x,1)+L\big(1-\nu'(x,1)\big)=\lambda_1\nu(x,1)-\lambda_1\nu(x,2),\\
\\
\frac{1}{2}\sigma^2(2)\nu''(x,2)+\mu(2)\nu'(x,2)-\delta\nu(x,2)+L\big(1-\nu'(x,2)\big)=\lambda_2\nu(x,2)-\lambda_2\nu(x,1).
\end{array}\right.
\end{eqnarray}
Considering the characteristic equation for \eqref{25}, $\phi_1^3(\widetilde\theta)\phi_2^3(\widetilde\theta)=\lambda_1\lambda_2$, where $\phi_1^3$ and $\phi_2^3$ have been defined in the previous section, and by the similar way we know the solution of the E.q.\eqref{25} is
\begin{eqnarray}
  \left\{\begin{array}{l}
  \label{26}
\nu(x,1)=\widetilde C_1e^{\widetilde\theta_1(x-B_1)}+\widetilde C_2e^{\widetilde\theta_2(x-B_1)}+\widetilde C_3e^{\widetilde\theta_3(x-B_1)}+\widetilde C_4e^{\widetilde\theta_4(x-B_1)}+\frac{L}{\delta},\\
\\
\nu(x,2)=\widetilde D_1e^{\widetilde\theta_1(x-B_1)}+\widetilde D_2e^{\widetilde\theta_2(x-B_1)}+\widetilde D_3e^{\widetilde\theta_3(x-B_1)}+\widetilde D_4e^{\widetilde\theta_4(x-B_1)}+\frac{L}{\delta},\\
\end{array}\right.
\end{eqnarray}
and for $j=1,2,3,4,$
\begin{eqnarray}
 \label{z5}
\widetilde D_j=\frac{\phi_1^3(\widetilde\theta_j)}{\lambda_1}\widetilde C_j=\frac{\lambda_2}{\phi_2^3(\widetilde\theta_j)}\widetilde C_j.
\end{eqnarray}
\\

When $x\in[B_1,B_2)$, $\nu'(x,1)=1$ and $\nu'(x,2)\leq1$, we have the following system of differential equations:
\begin{eqnarray}
  \left\{\begin{array}{l}
  \label{27}
\nu(x,1)=x-K,\\
\\
\frac{1}{2}\sigma^2(2)\nu''(x,2)+\mu(2)\nu'(x,2)-\delta\nu(x,2)+L\big(1-\nu'(x,2)\big)=\lambda_2\nu(x,2)-\lambda_2\nu(x,1).
\end{array}\right.
\end{eqnarray}
Consider the characteristic equation for the second equation above, $\phi^4(\hat\theta)=\frac{1}{2}\sigma_2^2\hat \theta^2+(\mu_2-L)\hat \theta-(\lambda_2+\delta)=0$. It has two real roots: $\hat\theta_1<\hat\theta_2$, then the solution of the ODE is:
\[\nu(x,2)=\hat D_1e^{\hat \theta_1(x-B_2)}+\hat D_2e^{\hat\theta_2(x-B_2)}+U(x),\]
where $U(x)$ has been defined in the previous section, that is
\[U(x)=\frac{\lambda_2}{\lambda_2+\delta}x+\frac{1}{\lambda_2+\delta}\bigg\{(\mu_2-L)\frac{\lambda_2}{\lambda_2+\delta}+L+\lambda_2\big[\nu(b_1,1)-b_1-K\big]\bigg\}.\]
\\

When $x\in[B_2,\infty)$, from the equation \eqref{a7}, we have
\begin{eqnarray}
  \left\{\begin{array}{l}
  \label{28}
\nu(x,1)=x-K,\\
\\
\nu(x,2)=\nu(b_2,2)+x-b_2-K,
\end{array}\right.
\end{eqnarray}
where $\nu(b_2,2)$ can be calculated by E.q.\eqref{26}.
\\

In order to find the thresholds $b_2,B_1,B_2$, and the coefficients $C_1,C_2,C_3,C_4$, $\widetilde C_1,\widetilde C_2,\widetilde C_3,\widetilde C_4$, $\hat D_1,\hat D_2$. We suppose the smooth fit condition hold, and combined with $\nu'(b_2,i)=1$ for $i=1,2$, we have the equations following:
\begin{eqnarray}
\label{a14}
\nonumber &&\nu(0,1)=0, \quad\nu(0,2)=0, \quad\nu(b_2+,1)=\nu(b_2-,1),\quad\nu(b_2+,2)=\nu(b_2-,2)\\
\nonumber &&\nu'(b_2+,1)=\nu'(b_2-,1), \quad \nu'(b_2+,2)=1, \quad \nu'(b_2-,2)=1\\
\nonumber &&\nu(B_1+,1)=\nu(B_1-,1), \quad \nu(B_1+,2)=\nu(B_1-,2), \quad \nu'(B_1+,2)=\nu'(B_1-,2)\\
&&\nu'(B_1-,1)=1,\quad\nu(B_2+,2)=\nu(B_2-,2), \quad \nu'(B_2-,2)=1.
\end{eqnarray}
Meanwhile the coefficients  $D_1,D_2,D_3,D_4$, $\widetilde D_1,\widetilde D_2,\widetilde D_3,\widetilde D_4$ can be obtained from equation \eqref{z4},\eqref{z5}.

Take notice that we assume $\nu'(0,1)=1$ and $\nu'(0,2)>1$. This condition is satisfied if and only if the coefficients found through the system of E.q.\eqref{24} satisfy:
\begin{eqnarray}
  \left\{\begin{array}{l}
  \label{29}
C_1\theta_1+C_2\theta_2+C_3\theta_3+C_4\theta_4=1,
\\
C_1\varphi_1^2(\theta_1) \theta_1+C_2\varphi_1^2(\theta_2)\theta_2+C_3\varphi_1^2(\theta_3)\theta_3+C_4\varphi_1^2(\theta_4)\theta_4>\lambda_1.
\end{array}\right.
\end{eqnarray}
\\

\textbf{Case 3}:  $\nu'(0,1)=1$, $\nu'(0,2)=1$.

There are three intervals need to be considered : $[0,B_1)$, $[B_1,B_2)$, $[B_2,\infty)$.

When $x\in[0,B_1)$, $\nu'(x,1)\leq1$ and $\nu'(x,2)\leq1$ then we have the following system of differential equations:
\begin{eqnarray}
  \left\{\begin{array}{l}
  \label{30}
\frac{1}{2}\sigma^2(1)\nu''(x,1)+\mu(1)\nu'(x,1)-\delta\nu(x,1)+L\big(1-\nu'(x,1)\big)=\lambda_1\nu(x,1)-\lambda_1\nu(x,2),\\
\\
\frac{1}{2}\sigma^2(2)\nu''(x,2)+\mu(2)\nu'(x,2)-\delta\nu(x,2)+L\big(1-\nu'(x,2)\big)=\lambda_2\nu(x,2)-\lambda_2\nu(x,1),
\end{array}\right.
\end{eqnarray}
Consider the characteristic equation for \eqref{30}, $\phi_1^3(\beta)\phi_2^3(\beta)=\lambda_1\lambda_2$, where $\phi_1^3$ and $\phi_2^3$ have been defined in case 1, and by the similar way we know the solution of the E.q.\eqref{30} is
\begin{eqnarray}
  \left\{\begin{array}{l}
  \label{31}
\nu(x,1)=M_1e^{\beta_1x}+M_2e^{\beta_2x}+M_3e^{\beta_3x}+M_4e^{\beta_4x}+\frac{L}{\delta},\\
\\
\nu(x,2)=N_1e^{\beta_1x}+N_2e^{\beta_2x}+N_3e^{N_3x}+N_4e^{\beta_4x}+\frac{L}{\delta},\\
\end{array}\right.
\end{eqnarray}
and for $j=1,2,3,4,$
\begin{eqnarray}
 \label{z6}
N_j=\frac{\phi_1^3(\beta_j)}{\lambda_1}M_j=\frac{\lambda_2}{\phi_2^3(\beta_j)}M_j,
\end{eqnarray}
where $\beta_1<\beta_2<0<\beta_3<\beta_4$ are characteristic roots.
\\

When $x\in[B_1,B_2)$, $\nu'(x,1)=1$ and $\nu'(x,2)\leq1$ then we have the following system of differential equations:
\begin{eqnarray}
  \left\{\begin{array}{l}
\nu(x,1)=x-K,\\
\\
\frac{1}{2}\sigma^2(2)\nu''(x,2)+\mu(2)\nu'(x,2)-\delta\nu(x,2)+L\big(1-\nu'(x,2)\big)=\lambda_2\nu(x,2)-\lambda_2\nu(x,1).
\end{array}\right.
\end{eqnarray}
Consider the characteristic equation for the second equation above, $\phi^4(\widetilde\beta)=\frac{1}{2}\sigma_2^2\widetilde\beta^2+(\mu_2-L)\widetilde\beta-(\lambda_2+\delta)=0$. It has two real roots: $\widetilde\beta_1<\widetilde\beta_2$, then the solution of the ODE is:
\[\nu(x,2)=\widetilde N_1e^{\widetilde\beta_1(x-B_2)}+\widetilde N_2e^{\widetilde\beta_2(x-B_2)}+U(x),\]
where $U(x)$ has been defined in the previous section.
\\

When $x\in[B_2,\infty)$. We have
\begin{eqnarray}
  \left\{\begin{array}{l}
  \label{32}
\nu(x,1)=x-K,\\
\\
\nu(x,2)=x-K.
\end{array}\right.
\end{eqnarray}
\\

In order to find the thresholds $B_1,B_2$, and the coefficients $M_1,M_2,M_3,M_4$, $\widetilde N_1,\widetilde N_2$. We suppose the smooth fit condition hold, thus we have the equations following:
\begin{eqnarray}
\label{a15}
\nonumber &&\nu(0,1)=0, \quad\nu(0,2)=0, \quad \nu(B_1-,1)=\nu(B_1+,1)\\
\nonumber &&\nu(B_1+,2)=\nu(B_1-,2), \quad \nu'(B_1+,2)=\nu'(B_1-,2), \quad \nu'(B_1-,1)=1\\
&&\nu(B_2+,2)=\nu(B_2-,2), \quad \nu'(B_2-,2)=1.
\end{eqnarray}
Meanwhile the coefficients  $N_1,N_2,N_3,N_4$ can be obtained from equation \eqref{z6}.

Note that we assume $\nu'(0,1)=1$ and $\nu'(0,2)=1$. This condition is satisfied if the coefficients found through the system of E.q.\eqref{31} satisfy:
\begin{eqnarray}
  \left\{\begin{array}{l}
  \label{32}
M_1\beta_1+M_2\beta_2+M_3\beta_3+M_4\beta_4=1,
\\
M_1\varphi_1^3(\beta_1)\beta_1+M_2\varphi_1^3(\beta_2)\beta_2+M_3\varphi_1^3(\beta_3)\beta_3+M_4\varphi_1^3(\beta_4)\beta_4=\lambda_1.
\end{array}\right.
\end{eqnarray}
\\

Next we will give a theorem which shows that $\nu$ is indeed the value function. We also give the optimal dividend policy.
\\
\begin{thm}
 Suppose that $b_1\leq b_2<B_1\leq B_2$\\

 Assume that $\dot A_j$, $\widetilde A_j$, $\hat A_j$, $j=1,2,3,4$, $\breve A_j$, $j=1,2$ be the solution of the system E.q.\eqref{a13} and suppose they satisfy condition \eqref{22}. In addition, $\dot B_j$, $\widetilde B_j$, $\hat B_j$, $j=1,2,3,4$, $\breve B_j$, $j=1,2$ and $U(x)$ are defined before. Then the function $\nu$ given by
 \begin{displaymath}
\nu(x,1) = \left\{ \begin{array}{ll}
\dot A_1e^{\dot \alpha_1x}+\dot A_2e^{\dot \alpha_2x}+\dot A_3e^{\dot \alpha_3x}+\dot A_4e^{\dot\alpha_4x}, & \textrm{if $x\in[0,b_1)$},\\
\widetilde A_1e^{\widetilde\alpha_1(x-b_2)}+\widetilde A_2e^{\widetilde\alpha_2(x-b_2)}+\widetilde A_3e^{\widetilde\alpha_3(x-b_2)}+\widetilde A_4e^{\widetilde\alpha_4(x-b_2)}+F_1, & \textrm{if $x\in[b_1,b_2)$},\\
\hat A_1e^{\hat\alpha_1(x-B_1)}+\hat A_2e^{\hat \alpha_2(x-B_1)}+\hat A_3e^{\hat\alpha_3(x-B_1)}+\hat A_4e^{\hat\alpha_4(x-B_1)}+\frac{L}{\delta}, & \textrm{if $x\in[b_2,B_1)$},\\
\widetilde A_1e^{\widetilde\alpha_1(b_1-b_2)}+\widetilde A_2e^{\widetilde\alpha_2(b_1-b_2)}+\widetilde A_3e^{\widetilde\alpha_3(b_1-b_2)}+\widetilde A_4e^{\widetilde\alpha_4(b_1-b_2)}+F_1+x-b_1-K, & \textrm{if $x\in[B_1,\infty)$},\\
\end{array} \right.
\end{displaymath}
and\\
 \begin{displaymath}
\nu(x,2) = \left\{ \begin{array}{ll}
\dot B_1e^{\dot \alpha_1x}+\dot B_2e^{\dot \alpha_2x}+\dot B_3e^{\dot \alpha_3x}+\dot B_4e^{\dot \alpha_4x}, & \textrm{if $x\in[0,b_1)$},\\
\widetilde B_1e^{\widetilde\alpha_1(x-b_2)}+\widetilde B_2e^{\widetilde\alpha_2(x-b_2)}+\widetilde B_3e^{\widetilde\alpha_3(x-b_2)}+\widetilde B_4e^{\widetilde\alpha_4(x-b_2)}+F_2, & \textrm{if $x\in[b_1,b_2)$},\\
\hat B_1e^{\hat\alpha_1(x-B_1)}+\hat B_2e^{\hat \alpha_2(x-B_1)}+\hat B_3e^{\hat\alpha_3(x-B_1)}+\hat B_4e^{\hat\alpha_4(x-B_1)}+\frac{L}{\delta}, & \textrm{if $x\in[b_2,B_1)$},\\
\breve B_1e^{\breve \alpha_1(x-B_2)}+\breve B_2e^{\breve\alpha_2(x-B_2)}+U(x), & \textrm{if $x\in[B_1,B_2)$},\\
\hat B_1e^{\hat\alpha_1(b_2-B_1)}+\hat B_2e^{\hat \alpha_2(b_2-B_1)}+\hat B_3e^{\hat\alpha_3(b_2-B_1)}+\hat B_4e^{\hat\alpha_4(b_2-B_1)}+\frac{L}{\delta}+x-b_2-K, & \textrm{if $x\in[B_2,\infty)$}\\
\end{array} \right.
\end{displaymath}
 are the value function. And optimal strategy $\hat \pi=\big(\hat u,\hat \Gamma,\hat \xi\big)=\pi^{\nu}=\big(u^\nu,\Gamma^\nu,\xi^\nu\big)$ is the QVI control associated with $\nu$, furthermore, $\hat u$ is defined by \\
 \begin{displaymath}
\hat u(t) = \left\{ \begin{array}{ll}
0 & \textrm{if $\epsilon_t=i$ and $X_t\in[0,b_i)$},\\
L & \textrm{if $\epsilon_t=i$ and $X_t\in[b_i,B_i)$},
\end{array} \right.
\end{displaymath}
for $t\in[0,\hat\Theta)$, and $\hat u(t)=0$ for $t\in[\hat\Theta,\infty)$.

Consider the case in which $\dot A_j$, $\widetilde A_j$, $\hat A_j$, $j=1,2,3,4$, $\breve A_j$, $j=1,2$ do not satisfy condition \eqref{22}. Instead suppose condition \eqref{29} are satisfied, and assume $C_j$, $\widetilde C_j$, $j=1,2,3,4$ and $\hat D_1$, $\hat D_2$ be the solution of the system E.q.\eqref{a14}. $D_j$, $\widetilde D_j$, $j=1,2,3,4$ and $F_j$, $j=1,2$ are defined before. Then the function $\nu$ given by
\begin{displaymath}
\nu(x,1) = \left\{ \begin{array}{ll}
C_1e^{\theta_1x}+C_2e^{\theta_2x}+C_3e^{\theta_3x}+C_4e^{\theta_4x}+F_1, & \textrm{if $x\in[0,b_2)$},\\
\widetilde C_1e^{\widetilde\theta_1(x-B_1)}+\widetilde C_2e^{\widetilde\theta_2(x-B_1)}+\widetilde C_3e^{\widetilde\theta_3(x-B_1)}+\widetilde C_4e^{\widetilde\theta_4(x-B_1)}+\frac{L}{\delta}, & \textrm{if $x\in[b_2,B_1)$},\\
x-K, & \textrm{if $x\in[B_1,\infty)$},\\
\end{array} \right.
\end{displaymath}
and
\begin{displaymath}
\nu(x,2) = \left\{ \begin{array}{ll}
D_1e^{\theta_1x}+D_2e^{\theta_2x}+D_3e^{\theta_3x}+D_4e^{\theta_4x}+F_2, & \textrm{if $x\in[0,b_2)$},\\
\widetilde D_1e^{\widetilde\theta_1(x-B_1)}+\widetilde D_2e^{\widetilde\theta_2(x-B_1)}+\widetilde D_3e^{\widetilde\theta_3(x-B_1)}+\widetilde D_4e^{\widetilde\theta_4(x-B_1)}+\frac{L}{\delta}, & \textrm{if $x\in[b_2,B_1)$},\\
\hat D_1e^{\hat \theta_1(x-B_2)}+\hat D_2e^{\hat\theta_2(x-B_2)}+U(x), & \textrm{if $x\in[B_1,B_2)$},\\
\widetilde D_1e^{\widetilde\theta_1(b_2-B_1)}+\widetilde D_2e^{\widetilde\theta_2(b_2-B_1)}+\widetilde D_3e^{\widetilde\theta_3(b_2-B_1)}+\widetilde D_4e^{\widetilde\theta_4(b_2-B_1)}+\frac{L}{\delta}+x-b_2-K, & \textrm{if $x\in[B_2,\infty)$}
\end{array} \right.
\end{displaymath}
are the value function. And optimal strategy $\hat \pi=\big(\hat u,\hat \Gamma,\hat \xi\big)=\pi^{\nu}=\big(u^\nu,\Gamma^\nu,\xi^\nu\big)$ is the QVI control associated with $\nu$, furthermore, $\hat u$ is defined by  \\
\begin{displaymath}
\hat u(t) = \left\{ \begin{array}{ll}
L & \textrm{if $\epsilon_t=1$ },\\
0 & \textrm{if $\epsilon_t=2$ and $X_t\in[0,b_2)$},\\
L & \textrm{if $\epsilon_t=2$ and $X_t\in[b_2,B_2)$},
\end{array} \right.
\end{displaymath}
for $t\in[0,\hat\Theta)$, and $\hat u(t)=0$ for $t\in[\hat \Theta,\infty)$.

Consider the case in which condition \eqref{22} and \eqref{29} are not satisfied, but condition \eqref{32} is satisfied. We suppose $M_j$,$j=1,2,3,4$; $\widetilde N_j$, $j=1,2$ be the solution of the system E.q.\eqref{a15}, and $N_j$, $j=1,2,3,4$, $U(x)$ are defined before. Then the function $\nu$ given by
\begin{displaymath}
\nu(x,1) = \left\{ \begin{array}{ll}
M_1e^{\beta_1x}+M_2e^{\beta_2x}+M_3e^{\beta_3x}+M_4e^{\beta_4x}+\frac{L}{\delta}, & \textrm{if $x\in[0,B_1)$},\\
x-K, & \textrm{if $x\in[B_1,\infty)$},
\end{array} \right.
\end{displaymath}
and
\begin{displaymath}
\nu(x,2) = \left\{ \begin{array}{ll}
N_1e^{\beta_1x}+N_2e^{\beta_2x}+N_3e^{\beta_3x}+N_4e^{\beta_4x}+\frac{L}{\delta}, & \textrm{if $x\in[0,B_1)$},\\
\widetilde N_1e^{\widetilde\beta_1(x-B_2)}+\widetilde N_2e^{\widetilde\beta_2(x-B_2)}+U(x), & \textrm{if $x\in[B_1,B_2)$},\\
x-K, & \textrm{if $x\in[B_2,\infty)$}
\end{array} \right.
\end{displaymath}
are the value function. And optimal strategy $\hat \pi=\big(\hat u,\hat \Gamma,\hat \xi\big)=\pi^{\nu}=\big(u^\nu,\Gamma^\nu,\xi^\nu\big)$ is the QVI control associated with $\nu$, furthermore, $\hat u$ defined by  \\
 \begin{displaymath}
\hat u(t) = \left\{ \begin{array}{ll}
L & \textrm{if $t\in[0,\hat\Theta)$},\\
0 & \textrm{if $t\in[\hat\Theta,\infty)$}.
\end{array} \right.
\end{displaymath}
\end{thm}
\begin{proof}
We omit the proof, because it is similar to and simpler than the proof of Theorem 4.5 below.
\end{proof}

\subsection{The case of $b_1<B_1<b_2<B_2$}
In this subsection we assume $b_1<B_1<b_2<B_2$. Considering the relationship between $\nu'(0,i)$, $i\in\{1,2\}$ and $1$, we have three cases: $\nu'(0,i)>1$ for both $i\in\{1,2\}$; $\nu'(0,i_0)=1$ and $\nu'(0,3-i_0)>1$ for some $i_0\in\{1,2\}$; and $\nu'(0,i)=1 $ for both $i\in\{1,2\}$.
\\

\textbf{Case 1}: $\nu'(0,1)>1$ and $\nu'(0,2)>1$ \\

According to the discussion above, we need to consider five possibilities: $x\in[0,b_1)$; $x\in[b_1,B_1)$; $x\in[B_1,b_2)$; $x\in[b_2,B_2)$ and $x\in[B_2,\infty)$.\\

When $x\in[0,b_1)$, we have $\nu'(x,1)>1$ and $\nu'(x,2)>1$, E.q.\eqref{a6} gives the following system of differential equations:

\begin{eqnarray}
  \left\{\begin{array}{l}
  \label{33}
\frac{1}{2}\sigma^2(1)\nu''(x,1)+\mu(1)\nu'(x,1)-\delta\nu(x,1)=\lambda_1\nu(x,1)-\lambda_1\nu(x,2),\\
\\
\frac{1}{2}\sigma^2(2)\nu''(x,2)+\mu(1)\nu'(x,2)-\delta\nu(x,2)=\lambda_2\nu(x,2)-\lambda_2\nu(x,1),
\end{array}\right.
\end{eqnarray}
which is the same as E.q.\eqref{13}. By the same way as in previous section, we have the solution:
\begin{eqnarray}
  \left\{\begin{array}{l}
  \label{34}
\nu(x,1)=\dot A_1e^{\dot \alpha_1x}+\dot A_2e^{\dot \alpha_2x}+\dot A_3e^{\dot\alpha_3x}+\dot A_4e^{\dot \alpha_4x},\\
\\
\nu(x,2)=\dot B_1e^{\dot \alpha_1x}+\dot B_2e^{\dot \alpha_2x}+\dot B_3e^{\dot\alpha_3x}+\dot B_4e^{\dot \alpha_4x},
\end{array}\right.
\end{eqnarray}
where for each $j=1,2,3,4$,\\
\begin{eqnarray}
\label{z7}
\dot B_j=\frac{\phi_1^1(\dot\alpha_j)}{\lambda_1}\dot A_j=\frac{\lambda_2}{\phi_2^1(\dot \alpha_j)}\dot A_j,
\end{eqnarray}
and $\dot \alpha_1<\dot \alpha_2<0<\dot \alpha_3<\dot \alpha_4$ are the four roots of characteristic equation $\phi_1^1(\dot \alpha)\phi_2^1(\dot \alpha)=\lambda_1\lambda_2$.\\

When $x\in[b_1,B_1)$, we have $\nu'(x,1)\leq1$ and $\nu'(x,2)>1$, E.q.\eqref{a5},\eqref{a6} give the following system of differential equations:
\begin{eqnarray}
  \left\{\begin{array}{l}
  \label{35}
\frac{1}{2}\sigma^2(1)\nu''(x,1)+\mu(1)\nu'(x,1)-\delta\nu(x,1)+L\big(1-\nu'(x,1)\big)=\lambda_1\nu(x,1)-\lambda_1\nu(x,2),\\
\\
\frac{1}{2}\sigma^2(2)\nu''(x,2)+\mu(1)\nu'(x,2)-\delta\nu(x,2)=\lambda_2\nu(x,2)-\lambda_2\nu(x,1).
\end{array}\right.
\end{eqnarray}
According to section 4.1, we know the solution of E.q.\eqref{35} is
\begin{eqnarray}
  \left\{\begin{array}{l}
  \label{36}
\nu(x,1)=\widetilde A_1e^{\widetilde\alpha_1(x-B_1)}+\widetilde A_2e^{\widetilde\alpha_2(x-B_1)}+\widetilde A_3e^{\widetilde\alpha_3(x-B_1)}+\widetilde A_4e^{\widetilde\alpha_4(x-B_1)}+F_1,\\
\\
\nu(x,2)=\widetilde B_1e^{\widetilde\alpha_1(x-B_1)}+\widetilde B_2e^{\widetilde\alpha_2(x-B_1)}+\widetilde B_3e^{\widetilde\alpha_3(x-B_1)}+\widetilde B_4e^{\widetilde\alpha_4(x-B_1)}+F_2,\\
\end{array}\right.
\end{eqnarray}
where $F_i:=\frac{(\lambda_2+(2-i)\delta)L}{(\lambda_i+\delta)(\lambda_2+\delta)-\lambda_1\lambda_2}$, $i=1,2$ and for $j=1,2,3,4,$
\begin{eqnarray}
\label{z8}
\widetilde B_j=\frac{\phi_1^2(\widetilde\alpha_j)}{\lambda_1}\widetilde A_j=\frac{\lambda_2}{\phi_2^2(\widetilde\alpha_j)}\widetilde A_j,
\end{eqnarray}
and  $\widetilde\alpha_1<\widetilde\alpha_2<0<\widetilde\alpha_3<\widetilde\alpha_4$ are the four roots of characteristic equation $\phi_1^2(\widetilde\alpha)\phi_2^2(\widetilde\alpha)=\lambda_1\lambda_2$.
\\

Next we give a theorem which is similar to the one in \cite{sotomayor2011classical}. To apply it in our situation, we need to adapt it in some details, and then we complete its proof.
\begin{thm}
If the solution $\nu$ of the system \eqref{35} is such that $\nu''(\cdot,1)$ is strictly increasing in certain left neighborhood of $B_1$, $\nu'(B_1-,1)=1$ and $\nu''(B_1-,1)\geq0$, then we have $\nu'(b_1-,2)<\frac{(\lambda_1+\delta)}{\lambda_1}$.
\end{thm}
\begin{proof}
If $\nu''(\cdot,1)$ is strictly increasing in certain left neighborhood of $B_1$, then exist constant $c$, for every $x\in[B_1-c,B_1)$, we have $\nu'''(x,1)>0$, in particular $\nu'''(B_1-,1)>0$. Hence, combined with the definition of $\phi_1^2$:
\begin{eqnarray*}
\lambda_1\nu'(B_1-,2)&=&\lambda_1\big[\widetilde B_1\widetilde\alpha_1+\widetilde B_2\widetilde\alpha_2+\widetilde B_3\widetilde\alpha_3+\widetilde B_4\widetilde\alpha_4\big]\\
&&= \lambda_1\big[\frac{\phi_1^2(\widetilde\alpha_1)}{\lambda_1}\widetilde A_1 \widetilde\alpha_1+\frac{\phi_1^2(\widetilde \alpha_2)}{\lambda_1}\widetilde A_2\widetilde\alpha_2+\frac{\phi_1^2(\widetilde \alpha_3)}{\lambda_1}\widetilde A_3 \widetilde\alpha_3+\frac{\phi_1^2(\widetilde \alpha_4)}{\lambda_1}\widetilde A_4 \widetilde\alpha_4\big]\\
&&=\phi_1^2(\widetilde\alpha_1)\widetilde A_1 \widetilde\alpha_1+\phi_1^2(\widetilde\alpha_2)\widetilde A_2 \widetilde\alpha_2+\phi_1^2(\widetilde\alpha_3)\widetilde A_3 \widetilde\alpha_3+\phi_1^2(\widetilde\alpha_4)\widetilde A_4 \widetilde\alpha_4\\
&&=-\frac{1}{2}\sigma_1^2\sum_{j=1}^4\widetilde\alpha_j^3\widetilde A_j-(\mu_1-L)\sum_{j=1}^4\widetilde\alpha_j^2\widetilde A_j+(\lambda_1+\delta)\sum_{j=1}^4\widetilde\alpha_j\widetilde A_j\\
&&=-\frac{1}{2}\sigma_1^2\nu'''(B_1-,1)-(\mu_1-L)\nu''(B_1-,1)+(\lambda_1+\delta)\nu'(B_1-,1)\leq\lambda_1+\delta.
\end{eqnarray*}
Here we use the assumption (H) $\mu_*>L$.
\end{proof}

When $x\in[B_1,b_2)$, $\nu'(x,1)=1$ and $\nu'(x,2)>1$, then we have the following system of differential equations:
\begin{eqnarray}
  \left\{\begin{array}{l}
  \label{37}
\nu(x,1)=\nu(b_1,1)+x-b_1-K,\\
\\
\frac{1}{2}\sigma^2(2)\nu''(x,2)+\mu(2)\nu'(x,2)-\delta\nu(x,2)=\lambda_2\nu(x,2)-\lambda_2\nu(x,1),
\end{array}\right.
\end{eqnarray}
where $\nu(b_1,1)$ can be calculated by E.q.\eqref{36}. Consider the characteristic equation for the second equation above. $\phi^5(\hat \alpha)=\frac{1}{2}\sigma_2^2\hat \alpha^2+\mu_2\hat\alpha-(\lambda_2+\delta)=0$. It has two real roots: $\hat\alpha_1<\hat\alpha_2$, then the solution of the ODE is:
\[\nu(x,2)=\hat B_1e^{\hat \alpha_1(x-b_2)}+\hat B_2e^{\hat\alpha_2(x-b_2)}+\bar U(x),\]
where $\bar U(x)=\frac{\lambda_2}{\lambda_2+\delta}x+\frac{1}{\lambda_2+\delta}\bigg\{\mu_2\frac{\lambda_2}{\lambda_2+\delta}+L+\lambda_2\big[\nu(b_1,1)-b_1-K\big]\bigg\}$. That is, the solution of E.q.\eqref{37} is
\begin{eqnarray}
  \left\{\begin{array}{l}
  \label{38}
\nu(x,1)=\nu(b_1,1)+x-b_1-K,\\
\\
\nu(x,2)=\hat B_1e^{\hat \alpha_1(x-b_2)}+\hat B_2e^{\hat\alpha_2(x-b_2)}+\bar U(x).
\end{array}\right.
\end{eqnarray}
\\

When $x\in[b_2,B_2)$, $\nu'(x,1)=1$ and $\nu'(x,2)<1$, then we have the following system of differential equations:
\begin{eqnarray}
  \left\{\begin{array}{l}
  \label{39}
\nu(x,1)=\nu(b_1,1)+x-b_1-K,\\
\\
\frac{1}{2}\sigma^2(2)\nu''(x,2)+\mu(2)\nu'(x,2)-\delta\nu(x,2)+L(1-\nu'(x,2))=\lambda_2\nu(x,2)-\lambda_2\nu(x,1),
\end{array}\right.
\end{eqnarray}
where $\nu(b_1,1)$ can be calculated by E.q.\eqref{36}. According to section 4.1, we know the solution of E.q.\eqref{39} is
\begin{eqnarray}
  \left\{\begin{array}{l}
  \label{40}
\nu(x,1)=\nu(b_1,1)+x-b_1-K,\\
\\
\nu(x,2)=\breve B_1e^{\breve \alpha_1(x-B_2)}+\breve B_2e^{\breve\alpha_2(x-B_2)}+U(x),
\end{array}\right.
\end{eqnarray}
where $U(x)=\frac{\lambda_2}{\lambda_2+\delta}x+\frac{1}{\lambda_2+\delta}\bigg\{(\mu_2-L)\frac{\lambda_2}{\lambda_2+\delta}+L+\lambda_2\big[\nu(b_1,1)-b_1-K\big]\bigg\}$, and $\breve\alpha_1<\breve\alpha_2$ are the two roots of the characteristic equation for the second equation above: $\phi^4(\breve \alpha)=\frac{1}{2}\sigma_2^2\breve \alpha^2+(\mu_2-L)\breve \alpha-(\lambda_2+\delta)=0$.\\

When $x\in[B_2,\infty)$. From the equation \eqref{9}, we have
\begin{eqnarray}
  \left\{\begin{array}{l}
  \label{41}
\nu(x,1)=\nu(b_1,1)+x-b_1-K,\\
\\
\nu(x,2)=\nu(b_2,2)+x-b_2-K,
\end{array}\right.
\end{eqnarray}
where $\nu(b_1,1)$ and $\nu(b_2,2)$ can be calculated by E.q.\eqref{36} and E.q.\eqref{40} respectively.\\

In order to find the thresholds $b_1,b_2,B_1,B_2$, and the coefficients $\dot A_1, \dot A_2, \dot A_3, \dot A_4$, $\widetilde A_1,\widetilde A_2,\widetilde A_3,\widetilde A_4$, $\breve B_1,\breve B_2$, $\hat B_1,\hat B_2$. We suppose the smooth fit condition hold, and combined with $\nu'(b_i,i)=1$ for $i=1,2$, we have the equations following:
\begin{eqnarray}
\label{a16}
\nonumber &&\nu(0,1)=0, \quad\nu(0,2)=0, \quad\nu(b_1+,1)=\nu(b_1-,1),\quad  \nu(b_1+,2)=\nu(b_1-,2),\\
\nonumber &&\nu'(b_1+,1)=1, \quad\nu'(b_1-,1)=1, \quad\nu'(b_1+,2)=\nu'(b_1-,2), \quad\nu(b_2+,2)=\nu(b_2-,2),\\
\nonumber &&\nu'(b_2+,2)=1, \quad \nu'(b_2-,2)=1,\quad \nu(B_1+,1)=\nu(B_1-,1),\\
\nonumber &&\nu(B_1+,2)=\nu(B_1-,2), \quad \nu'(B_1+,2)=\nu'(B_1-,2), \quad \nu'(B_1-,1)=1,\\
&&\nu(B_2+,2)=\nu(B_2-,2), \quad \nu'(B_2-,2)=1.
\end{eqnarray}
Meanwhile the coefficients  $\dot B_1, \dot B_2, \dot B_3, \dot B_4$, $\widetilde B_1,\widetilde B_2,\widetilde B_3,\widetilde B_4$ can be obtained from equation \eqref{z7},\eqref{z8}.

Note that we assume $\nu'(0,1)>1$ and $\nu'(0,2)>1$. This condition is satisfied if the coefficients found through the system of E.q.\eqref{34} satisfy:
\begin{eqnarray}
  \left\{\begin{array}{l}
  \label{42}
\dot A_1\dot\alpha_1+\dot A_2\dot\alpha_2+\dot A_3\dot\alpha_3+\dot A_4\dot\alpha_4>1,
\\
\dot A_1\varphi_1^1(\dot\alpha_1) \alpha_1+\dot A_2\varphi_1^1(\dot\alpha_2)\alpha_2+\dot A_3\varphi_1^1(\dot\alpha_3)\alpha_3+\dot A_4\varphi_1^1(\dot\alpha_4)\alpha_4>\lambda_1.
\end{array}\right.
\end{eqnarray}
\\

\textbf{Case 2}: $\nu'(0,1)=1$ and $\nu'(0,2)>1$ \\

According to the discussion above, we need to consider four possibilities: $x\in[0,B_1)$; $x\in[B_1,b_2)$; $x\in[b_2,B_2)$ and $x\in[B_2,\infty)$.\\

When $x\in[0,B_1)$, we have $\nu'(x,1)\leq1$ and $\nu'(x,2)>1$, E.q.\eqref{a5},\eqref{a6} give the following system of differential equations:
\begin{eqnarray}
  \left\{\begin{array}{l}
  \label{43}
\frac{1}{2}\sigma^2(1)\nu''(x,1)+\mu(1)\nu'(x,1)-\delta\nu(x,1)+L(1-\nu'(x,1))=\lambda_1\nu(x,1)-\lambda_1\nu(x,2),\\
\\
\frac{1}{2}\sigma^2(2)\nu''(x,2)+\mu(1)\nu'(x,2)-\delta\nu(x,2)=\lambda_2\nu(x,2)-\lambda_2\nu(x,1).
\end{array}\right.
\end{eqnarray}
By the similar way, we know the solution of E.q.\eqref{43} is
\begin{eqnarray}
  \left\{\begin{array}{l}
  \label{44}
\nu(x,1)= C_1e^{\theta_1(x-B_1)}+ C_2e^{\theta_2(x-B_1)}+ C_3e^{\theta_3(x-B_1)}+ C_4e^{\theta_4(x-B_1)}+F_1,\\
\\
\nu(x,2)= D_1e^{\theta_1(x-B_1)}+ D_2e^{\theta_2(x-B_1)}+ D_3e^{\theta_3(x-B_1)}+ D_4e^{\theta_4(x-B_1)}+F_2,\\
\end{array}\right.
\end{eqnarray}
where $F_i:=\frac{(\lambda_2+(2-i)\delta)L}{(\lambda_i+\delta)(\lambda_2+\delta)-\lambda_1\lambda_2}$, $i=1,2$ and for $j=1,2,3,4,$
\begin{eqnarray}
\label{z9}
D_j=\frac{\phi_1^2(\theta_j)}{\lambda_1} C_j=\frac{\lambda_2}{\phi_2^2(\theta_j)} C_j,
\end{eqnarray}
and  $\theta_1<\theta_2<0<\theta_3<\theta_4$ are the four roots of characteristic equation $\phi_1^2(\theta)\phi_2^2(\theta)=\lambda_1\lambda_2$.
\\

When $x\in[B_1,b_2)$, $\nu'(x,1)=1$ and $\nu'(x,2)>1$, then we have the following system of differential equations:
\begin{eqnarray}
  \left\{\begin{array}{l}
  \label{45}
\nu(x,1)=\nu(b_1,1)+x-b_1-K=x-K,\\
\\
\frac{1}{2}\sigma^2(2)\nu''(x,2)+\mu(2)\nu'(x,2)-\delta\nu(x,2)=\lambda_2\nu(x,2)-\lambda_2\nu(x,1).
\end{array}\right.
\end{eqnarray}
Consider the characteristic equation for the second equation above: $\phi^5(\widetilde \theta)=\frac{1}{2}\sigma_2^2\widetilde \theta^2+\mu_2\widetilde \theta-(\lambda_2+\delta)=0$. It has two real roots: $\widetilde \theta_1<\widetilde \theta_2$, then the solution of the ODE is:
\[\nu(x,2)=\widetilde D_1e^{\widetilde \theta_1(x-b_2)}+\widetilde D_2e^{\widetilde \theta_2(x-b_2)}+\bar U(x),\]
where $\bar U(x)=\frac{\lambda_2}{\lambda_2+\delta}x+\frac{1}{\lambda_2+\delta}\big\{\mu_2\frac{\lambda_2}{\lambda_2+\delta}+L-\lambda_2 K\big\}$. That is the solution of E.q.\eqref{45} is
\begin{eqnarray}
  \left\{\begin{array}{l}
  \label{46}
\nu(x,1)=x-K,\\
\\
\nu(x,2)=\widetilde D_1e^{\widetilde \theta_1(x-b_2)}+\widetilde D_2e^{\widetilde \theta_2(x-b_2)}+\bar U(x).
\end{array}\right.
\end{eqnarray}
\\

When $x\in[b_2,B_2)$, we have $\nu'(x,1)=1$ and $\nu'(x,2)\leq1$. E.q.\eqref{a5},\eqref{a7} give the following system of differential equations:
\begin{eqnarray}
  \left\{\begin{array}{l}
  \label{47}
\nu(x,1)=\nu(b_1,1)+x-b_1-K=x-K,\\
\\
\frac{1}{2}\sigma^2(2)\nu''(x,2)+\mu(2)\nu'(x,2)-\delta\nu(x,2)+L(1-\nu'(x,2))=\lambda_2\nu(x,2)-\lambda_2\nu(x,1).
\end{array}\right.
\end{eqnarray}
Consider the characteristic equation for the second equation above. $\phi^4(\hat \theta)=\frac{1}{2}\sigma_2^2\hat \theta^2+(\mu_2-L)\hat \theta-(\lambda_2+\delta)=0$. It has two real roots: $\hat \theta_1<\hat \theta_2$, then the solution of the ODE is:
\[\nu(x,2)=\hat D_1e^{\hat \theta_1(x-B_2)}+\hat D_2e^{\hat \theta_2(x-B_2)}+U(x),\]
where $U(x)=\frac{\lambda_2}{\lambda_2+\delta}x+\frac{1}{\lambda_2+\delta}\big\{(\mu_2-L)\frac{\lambda_2}{\lambda_2+\delta}+L-\lambda_2 K\big\}$.
That is the E.q.\eqref{47} has solution:
\begin{eqnarray}
  \left\{\begin{array}{l}
  \label{48}
\nu(x,1)=x-K,\\
\\
\nu(x,2)=\hat D_1e^{\hat \theta_1(x-B_2)}+\hat D_2e^{\hat \theta_2(x-B_2)}+U(x).
\end{array}\right.
\end{eqnarray}
\\

When $x\in[B_2,\infty)$. From the equation \eqref{a7}, we have
\begin{eqnarray}
  \left\{\begin{array}{l}
  \label{49}
\nu(x,1)=x-K,\\
\\
\nu(x,2)=\nu(b_2,2)+x-b_2-K,
\end{array}\right.
\end{eqnarray}
where $\nu(b_2,2)$ can be calculated by the result in the case of $x\in[b_2,B_2)$.
\\

In order to find the thresholds $b_2,B_1,B_2$, and the coefficients $C_1,C_2,C_3,C_4$, $\widetilde D_1,\widetilde D_2 $, $\hat D_1,\hat D_2$. We suppose the smooth fit condition hold and $\nu'(b_i,i)=1$ for $i=1,2$, thus we have the equations following:
\begin{eqnarray}
\label{a17}
\nonumber &&\nu(0,1)=0, \quad\nu(0,2)=0,\quad \nu(b_2+,2)=\nu(b_2-,2), \quad \nu'(b_2+,2)=1, \\
\nonumber &&\nu'(b_2-,2)=1,\quad \nu(B_1+,1)=\nu(B_1-,1),\quad \nu(B_1+,2)=\nu(B_1-,2), \quad \nu'(B_1+,2)=\nu'(B_1-,2), \\
&&\nu'(B_1-,1)=1,\quad \nu(B_2+,2)=\nu(B_2-,2), \quad \nu'(B_2-,2)=1.
\end{eqnarray}
Meanwhile the coefficients  $D_1,D_2,D_3,D_4$ can be obtained from equation \eqref{z8}.

Note that we assume $\nu'(0,1)=1$ and $\nu'(0,2)>1$. This condition is satisfied if the coefficients found through the system of E.q.\eqref{46} satisfy:
\begin{eqnarray}
  \left\{\begin{array}{l}
  \label{50}
C_1\theta_1+C_2\theta_2+C_3\theta_3+C_4\theta_4>1,
\\
C_1\varphi_1^2(\theta_1) \theta_1+C_2\varphi_1^2(\theta_2)\theta_2+C_3\varphi_1^2(\theta_3)\theta_3+C_4\varphi_1^2(\theta_4)\theta_4>\lambda_1.
\end{array}\right.
\end{eqnarray}
\\

\begin{thm}
 Suppose that $b_1<B_1<b_2<B_2$\\

 Assume that $\dot A_j$, $\widetilde A_j$, $j=1,2,3,4$, $\breve B_j$, $\hat B_j$, $j=1,2$ be the solution of the system E.q.\eqref{a16} and suppose they satisfy condition \eqref{42}. Assume that  $\dot B_j$, $\widetilde B_j$, $j=1,2,3,4$ and $U(x)$, $\bar U(x)$are defined before. Then the function $\nu$ given by
 \begin{displaymath}
\nu(x,1) = \left\{ \begin{array}{ll}
\dot A_1e^{\alpha_1x}+\dot A_2e^{\alpha_2x}+\dot A_3e^{\alpha_3x}+\dot A_4e^{\alpha_4x}, & \textrm{if $x\in[0,b_1)$},\\
\widetilde A_1e^{\widetilde\alpha_1(x-B_1)}+\widetilde A_2e^{\widetilde\alpha_2(x-B_1)}+\widetilde A_3e^{\widetilde\alpha_3(x-B_1)}+\widetilde A_4e^{\widetilde\alpha_4(x-B_1)}+F_1, & \textrm{if $x\in[b_1,B_1)$},\\
\widetilde A_1e^{\widetilde\alpha_1(b_1-B_1)}+\widetilde A_2e^{\widetilde\alpha_2(b_1-B_1)}+\widetilde A_3e^{\widetilde\alpha_3(b_1-B_1)}+\widetilde A_4e^{\widetilde\alpha_4(b_1-B_1)}+F_1+x-b_1-K, & \textrm{if $x\in[B_1,\infty)$},\\
\end{array} \right.
\end{displaymath}
and\\
 \begin{displaymath}
\nu(x,2) = \left\{ \begin{array}{ll}
\dot B_1e^{\alpha_1x}+\dot B_2e^{\alpha_2x}+\dot B_3e^{\alpha_3x}+\dot B_4e^{\alpha_4x}, & \textrm{if $x\in[0,b_1)$},\\
\widetilde B_1e^{\widetilde\alpha_1(x-B_1)}+\widetilde B_2e^{\widetilde\alpha_2(x-B_1)}+\widetilde B_3e^{\widetilde\alpha_3(x-B_1)}+\widetilde B_4e^{\widetilde\alpha_4(x-B_1)}+F_2, & \textrm{if $x\in[b_1,B_1)$},\\
\hat B_1e^{\hat \alpha_1(x-b_2)}+\hat B_2e^{\hat\alpha_2(x-b_2)}+\bar U(x), & \textrm{if $x\in[B_1,b_2)$},\\
\breve B_1e^{\breve \alpha_1(x-B_2)}+\breve B_2e^{\breve\alpha_2(x-B_2)}+U(x), & \textrm{if $x\in[b_2,B_2)$},\\
\breve B_1e^{\breve \alpha_1(b_2-B_2)}+\breve B_2e^{\breve\alpha_2(b_2-B_2)}+U(b_2)+x-b_2-K ,& \textrm{if $x\in[B_2,\infty)$}\\
\end{array} \right.
\end{displaymath}
is the value function. And optimal strategy $\hat \pi=\big(\hat u,\hat \Gamma,\hat \xi\big)=\pi^{\nu}=\big(u^\nu,\Gamma^\nu,\xi^\nu\big)$ is the QVI control associated with $\nu$, furthermore, $\hat u$ defined by \\
 \begin{displaymath}
\hat u(t) = \left\{ \begin{array}{ll}
0 & \textrm{if $\epsilon_t=i$ and $X_t\in[0,b_i)$},\\
L & \textrm{if $\epsilon_t=i$ and $X_t\in[b_i,B_i)$},
\end{array} \right.
\end{displaymath}
for $t\in[0,\hat\Theta)$, and $\hat u(t)=0$ for $t\in[\hat\Theta,\infty)$.

Consider the case in which $\dot A_j$, $\widetilde A_j$, $j=1,2,3,4$, $\hat A_j$, $\breve A_j$, $j=1,2$ do not satisfy condition \eqref{42}. Instead suppose condition \eqref{50} are satisfied, and assume $C_j$, $j=1,2,3,4$ and $\hat D_j$, $\widetilde D_j$, $j=1,2$ be the solution of the system E.q.\eqref{a17}. $D_j$, $j=1,2,3,4$ and $F_j$, $j=1,2$, $U(x)$, $\bar U(x)$ are defined before. Then the function $\nu$ given by
 \begin{displaymath}
\nu(x,1) = \left\{ \begin{array}{ll}
C_1e^{\theta_1(x-B_1)}+C_2e^{\theta_2(x-B_1)}+C_3e^{\theta_3(x-B_1)}+C_4e^{\theta_4(x-B_1)}+F_1, & \textrm{if $x\in[0,B_1)$},\\
x-K, & \textrm{if $x\in[B_1,\infty)$},\\
\end{array} \right.
\end{displaymath}
and
\begin{displaymath}
\nu(x,2) = \left\{ \begin{array}{ll}
D_1e^{\theta_1(x-B_1)}+D_2e^{\theta_2(x-B_1)}+D_3e^{\theta_3(x-B_1)}+D_4e^{\theta_4(x-B_1)}+F_2, & \textrm{if $x\in[0,B_1)$},\\
\widetilde D_1e^{\widetilde\theta_1(x-b_2)}+\widetilde D_2e^{\widetilde\theta_2(x-b_2)}+\bar U(x), & \textrm{if $x\in[B_1,b_2)$},\\
\hat D_1e^{\hat\theta_1(x-B_2)}+\hat D_2e^{\hat\theta_2(x-B_2)}+U(x), & \textrm{if $x\in[b_2,B_2)$},\\
\hat D_1e^{\hat\theta_1(b_2-B_2)}+\hat D_2e^{\hat\theta_2(b_2-B_2)}+U(b_2)+x-b_2-K, & \textrm{if $x\in[B_2,\infty)$}\\
\end{array} \right.
\end{displaymath}
is the value function. And optimal strategy $\hat \pi=\big(\hat u,\hat \Gamma,\hat \xi\big)=\pi^{\nu}=\big(u^\nu,\Gamma^\nu,\xi^\nu\big)$ is the QVI control associated with $\nu$, furthermore, $\hat u$ defined by \\
\begin{displaymath}
\hat u(t) = \left\{ \begin{array}{ll}
L & \textrm{if $\epsilon_t=1$},\\
0 & \textrm{if $\epsilon_t=2$ and $X_t\in[0,b_2)$},\\
L & \textrm{if $\epsilon_t=2$ and $X_t\in[b_2,B_2)$},
\end{array} \right.
\end{displaymath}
for $t\in[0,\hat\Theta)$, and $\hat u(t)=0$ for $t\in[\hat \Theta,\infty)$.
\end{thm}

\begin{proof}
To prove the function $\nu$ defined above is value function, we have only to show it satisfies the conditions of theorem 3.3. Obviously, we have $\nu(x,1)\in C^2([0,\infty)/\{b_1,b_2,B_1\})$ and $\nu(x,2)\in C^2([0,\infty)/\{b_1,b_2,B_1,B_2\})$. From smooth fit condition, we can see $\nu(0,i)=0$, $i\in\{1,2\}$, and $\nu(0,i)$, $i\in\{1,2\}$ is continuously differentiable function. On the other hand, it is manifest that $\nu(x,i)$ is linear on $[B_i,\infty)$, $i=\{1,2\}$. In the following, if we can prove $\nu(x,i)$, $i\in\{1,2\}$ satisfy QVI, then we would complete the proof.

Firstly, we consider the case in which coefficients $\dot A_j$, $\widetilde A_j$, $j=1,2,3,4$, $\breve A_j$, $\hat A_j$, $j=1,2$ are the solution of the system E.q.\eqref{a16} and they satisfy condition  \eqref{42}, that is, the case of $\nu'(0,i)>1$, for $i\in\{1,2\}$. Taking account of $\nu'(x,1)$, on the interval $[0,b_1)$, $\nu'(x,1)>1$, then we have
\[\widetilde L_1(u)\nu(x,1)+u\leq\frac{1}{2}\sigma_1^2\nu''(x,1)+\mu_1\nu'(x,1)-(\lambda_1+\delta)\nu(x,1)+\lambda_1\nu(x,2)=0,\]
where $\widetilde L_i(u)\psi=\frac{1}{2}\sigma_i^2\psi''(\cdot,i)+\big[\mu(i)-u\big]\psi'(\cdot,i)-\delta\psi(\cdot,i)-\lambda_i\psi(\cdot,i)+\lambda_i\psi(\cdot,2-i)$, $i\in \{1,2\}$
and we have
\[\max_{u\in[0,L]}\bigg\{\widetilde L_1(u)\nu(x,1)+u\bigg\}=\widetilde L_1(\hat u)\nu(x,1)+\hat u=0.\]
Then we have equation
\[\big(\nu(x,1)-\mathcal{M}\nu(x,1)\big)\big(\max_{u\in[0,L]}\bigg\{\widetilde L_1(u)\nu(x,1)+u\bigg\}\big)=0.\]
 Further, we can prove, on the interval $[0,b_1)$, inequality $\nu(x,1)>\mathcal{M}\nu(x,1)$ holds. Indeed by differentiating $\nu(x-\eta)+\eta-K$ with respect to $\eta$, we can see function $\nu(x-\eta)+\eta-K$ is decrease with respect to $\eta$. Thus on the interval $[0,b_1)$
\[\mathcal{M}\nu(x,1)=\lim_{\eta\mapsto0^+}{\nu(x-\eta)+\eta-K}=\nu(x,1)-K<\nu(x,1).\]
So QVI is satisfied for $\nu(x,1)$ on interval $[0,b_1)$.

Secondly, on the interval $[b_1,B_1)$, $\nu'(x,1)\leq1$, then we have
\[\widetilde L_1(u)\nu(x,1)+u\leq\frac{1}{2}\sigma_1^2\nu''(x,1)+(\mu_1-L)\nu'(x,1)-(\lambda_1+\delta)\nu(x,1)+\lambda_1\nu(x,2)+L=0,\]
and
\[\max_{u\in[0,L]}\bigg\{\widetilde L_1(u)\nu(x,1)+u\bigg\}=\widetilde L_1(\hat u)\nu(x,1)+\hat u=0.\]
On the other hand, we have
\[\nu(x,1)>\mathcal{M}\nu(x,1),\]
indeed differentiating $\nu(x-\eta)+\eta-K$ with respect to $\eta$, leads to $\frac{\partial[\nu(x-\eta)+\eta-K]}{\partial\eta}\geq0$, which imply when $x-\eta\in[b_1,B_1)$, function $\nu(x-\eta)+\eta-K$ is increase with respect to $\eta$ and at the point $\eta=x-b_1$, the maximum is obtain. Thus when $x\in[b_1,B_1)$,
\begin{eqnarray*}
\mathcal{M}\nu(x,1)&=&\nu(b_1,1)+x-b_1-K\\
&=&\nu(b_1,1)+B_1-b_1-(B_1-x)-K\\
&=&\nu(B_1,1)-(B_1-x)\leq\nu(x,1).
\end{eqnarray*}

Thirdly, we consider the case $x\in[B_1,b_2)$. To begin with, we need to note $\nu'(x,2)$ is decrease on the interval $[B_1,b_2)$. Indeed, for $x\in[B_1,b_2)$
\[\nu'(x,2)=\hat\alpha_1 \hat A_1e^{\hat \alpha_1(x-b_2)}+\hat\alpha_2 \hat A_2e^{\hat\alpha_2(x-b_2)}+\frac{\lambda_2}{\lambda_2+\delta}.\]
From definition of $b_2$ and smooth fit condition, we get $\nu'(x,2)>\nu'(b_2,2)=1$, then we have
\[\hat\alpha_1 \hat A_1e^{\hat \alpha_1(x-b_2)}+\hat\alpha_2 \hat A_2e^{\hat\alpha_2(x-b_2)}+\frac{\lambda_2}{\lambda_2+\delta}>1,\]
thus
\[\hat\alpha_1 \hat A_1e^{\hat \alpha_1(x-b_2)}+\hat\alpha_2 \hat A_2e^{\hat\alpha_2(x-b_2)}>\frac{\delta}{\lambda_2+\delta}>0,\]
hence
\[\nu'''(x,2)=\hat\alpha_1^2\hat\alpha_1\hat A_1e^{\hat \alpha_1(x-b_2)}+\hat\alpha_1^2\hat\alpha_2\hat A_2e^{\hat\alpha_2(x-b_2)}>0,\]
which imply $\nu'(x,2)$ is strictly convex function. Combing with the truth $\nu'(x,2)>\nu'(b_2,2)$. we know $\nu'(x,2)$ strictly decrease on the interval $[B_1,b_2)$.

Next we need to note $\nu''(x,1)$ is strictly increasing in certain left neighborhood of point $B_1$. Indeed, by smooth fit condition and continuity, function $\nu'(B_1-,1)>0$, therefore on this neighborhood we have
\[\nu'(x,1)=\widetilde \alpha_1 \widetilde A_1e^{\widetilde\alpha_1(x-B_1)}+\widetilde \alpha_2\widetilde A_2e^{\widetilde\alpha_2(x-B_1)}+\widetilde \alpha_3\widetilde A_3e^{\widetilde\alpha_3(x-B_1)}+\widetilde \alpha_4\widetilde A_4e^{\widetilde\alpha_4(x-B_1)}>0,\]
then we obtain
\[\nu'''(x,1)=\widetilde \alpha_1^3 \widetilde A_1e^{\widetilde\alpha_1(x-B_1)}+\widetilde \alpha_2^3\widetilde A_2e^{\widetilde\alpha_2(x-B_1)}+\widetilde \alpha_3^3\widetilde A_3e^{\widetilde\alpha_3(x-B_1)}+\widetilde \alpha_4^3\widetilde A_4e^{\widetilde\alpha_4(x-B_1)}>0.\]
As a result we know $\nu''(x,1)$ is strictly increasing on this neighborhood.

Now we begin to prove that, on the interval $[B_1,b_2)$, the function $\nu(x,1)$ satisfies QVI. For $x\in[B_1,b_2)$ and for every $u\in[0,L]$,
\begin{eqnarray*}
\mathcal{A}(x)&:=&\widetilde L_1(u)\nu(x,1)+u=\frac{1}{2}\sigma_1^2\nu''(x,1)+(\mu_1-u)\nu'(x,1)-(\delta+\lambda_1)\nu(x,1)+\lambda_1\nu(x,2)+u\\
&=&(\mu_1-u)-(\delta+\lambda_1)\nu(x,1)+\lambda_1\nu(x,2)+u\\
&=&\mu_1-(\delta+\lambda_1)\nu(x,1)+\lambda_1\nu(x,2),
\end{eqnarray*}
then we have
\[\mathcal{A}'(x)=-(\delta+\lambda_1)+\lambda_1\nu'(x,2),\]
and
\begin{eqnarray}
\label{b1} \mathcal{A}''(x)=\lambda_1\nu''(x,2)<0.
\end{eqnarray}
Here we use the conclusion that $\nu'(x,2)$ is strictly decreased on the interval $[B_1,b_2)$.
\eqref{b1} imply $\mathcal{A}(x)$ is concave function on the interval $[B_1,b_2)$, and lead to $\mathcal{A}'(x)$ is decrease on $[B_1,b_2)$, then combining with theorem 4.4, we have
\begin{eqnarray*}
\mathcal{A}'(x)&\leq&\mathcal{A}'(B_1)=-(\delta+\lambda_1)+\lambda_1\nu'(B_1,2)=-(\delta+\lambda_1)+\lambda_1\nu'(B_1-,2)\\
&\leq&-(\delta+\lambda_1)+\lambda_1\frac{\delta+\lambda_1}{\lambda_1}=0.
\end{eqnarray*}
We know $\mathcal{A}(x)$ is decreased, then we have $\mathcal{A}(x)\leq\mathcal{A}(B_1)\leq0$

On the other hand, we can prove for $x\in[B_1,b_2)$, $\nu(x,1)=\mathcal{M}\nu(x,1)$. Indeed for $x\in[B_1,b_2)$,
\[\nu(x,1)=\nu(b_1,1)+x-b_1-K\]
let $\eta^*=x-b_1$
\[\nu(x-\eta^*,1)+\eta^*-K=\nu(b_1,1)+x-b_1-K.\]
so we get $\nu(x,1)=\mathcal{M}\nu(x,1)$. Combining with the obvious truth $\nu(x,1)\geq\mathcal{M}\nu(x,1)$, we have completed the proof that $\nu(x,1)$ satisfies QVI on the $[B_1,b_2)$.

In the end we verify that on $[b_2,\infty)$, $\nu(x,1)$ satisfies QVI. With the same argument, we know for $x\in[b_2,\infty)$, $\nu(x,1)=\mathcal{M}\nu(x,1)$. Besides, we have
\[\max_{u\in[0,L]}\bigg\{\widetilde L_1(u)\nu(x,1)+u\bigg\}\leq0.\]
Indeed, since $\nu''(b_2,1)=0$ and the result satisfied by $\nu(x,1)$ on the interval $ x\in[B_1,b_2)$. We have
\begin{eqnarray*}
&&(\mu_1-u)\nu'(b_2,1)-(\delta+\lambda_1)\nu(b_2,1)+\lambda_1\nu(b_2,2)+u\\
&&=\frac{1}{2}\sigma_1^2\nu''(b_2-,1)+(\mu_1-u)\nu'(b_2,1)-(\delta+\lambda_1)\nu(b_2,1)+\lambda_1\nu(b_2,2)+u\\
&&\leq0.
\end{eqnarray*}
So for $x>b_2$, note $\nu''(x,1)=0$, then we have
\begin{eqnarray*}
&&\frac{1}{2}\sigma_1^2\nu''(x,1)+(\mu_1-u)\nu'(x,1)-(\delta+\lambda_1)\nu(x,1)+\lambda_1\nu(x,2)+u\\
&&=(\mu_1-u)-(\delta+\lambda_1)\nu(x,1)+\lambda_1\nu(x,2)+u\\
&&=(\mu_1-u)-(\delta+\lambda_1)\nu(b_2,1)+\lambda_1\nu(b_2,2)+u+(\delta+\lambda_1)[\nu(b_2,1)-\nu(x,1)]+\lambda_1[\nu(x,2)-\nu(b_2,2)]\\
&&\leq(\delta+\lambda_1)\big[\nu(b_2,1)-\nu(x,1)\big]+\lambda_1\big[\nu(x,2)-\nu(b_2,2)\big]\\
&&=(\delta+\lambda_1)\big[\nu(b_1,1)+b_2-b_1-K-(\nu(b_1,1)+x-b_1-K)\big]+\lambda_1\big[\nu(x,2)-\nu(b_2,2)\big]\\
&&=(\delta+\lambda_1)(b_2-x)+\lambda_1\nu'(\xi,2)(x-b_2)\\
&&\leq(\delta+\lambda_1)(b_2-x)+\lambda_1(x-b_2)\\
&&\leq-\delta(x-b_2)\leq0,
\end{eqnarray*}
where $\xi\in[b_2,x)$, which lead to $\nu'(\xi,2)\leq1$. Then we have
\[\max_{u\in[0,L]}\bigg\{\widetilde L_1(u)\nu(x,1)+u\bigg\}\leq0.\]

In the following, we need to verify function $\nu(x,2)$ satisfies QVI, which is similar to the way of proving the case of $\nu(x,1)$. We demonstrate briefly:

At first, on the interval $[0,b_2)$, $\nu'(x,2)>1$, we have
\[\widetilde L_2(u)\nu(x,2)+u\leq\frac{1}{2}\sigma_1^2\nu''(x,2)+\mu_1\nu'(x,2)-(\lambda_2+\delta)\nu(x,1)+\lambda_2(x,2)=0,\]
and
\[\max_{u\in[0,L]}\bigg\{\widetilde L_2(u)\nu(x,2)+u\bigg\}=\widetilde L_2(\hat u)\nu(x,2)+\hat u=0.\]
Further, on the interval $[0,b_2)$, inequality $\nu(x,2)>\mathcal{M}\nu(x,2)$ hold. As a result, we know QVI establish on $[0,b_2)$

Secondly, on the interval $[b_2,B_2)$, $\nu'(x,2)\leq1$, we have
\[\widetilde L_2(u)\nu(x,2)+u\leq\frac{1}{2}\sigma_1^2\nu''(x,2)+(\mu_1-L)\nu'(x,2)-(\lambda_2+\delta)\nu(x,1)+\lambda_2\nu(x,2)+L=0,\]
and
\[\max_{u\in[0,L]}\bigg\{\widetilde L_2(u)\nu(x,2)+u\bigg\}=\widetilde L_2(\hat u)\nu(x,2)+\hat u=0.\]
On the other hand, inequality $\nu(x,2)>\mathcal{M}\nu(x,2)$ hold on $[b_2,B_2)$. Then we know QVI established on this interval.

Thirdly, we need to verify QVI hold for $\nu(x,2)$ on $[B_2,\infty)$. Since $\nu''(B_2-,1)\geq0$ and $\nu'(B_1,1)=1 $, then we have
\begin{eqnarray*}
&&(\mu_2-u)\nu'(B_2,1)-(\delta+\lambda_2)\nu(B_2,1)+\lambda_2\nu(B_2,2)+u\\
&&=\frac{1}{2}\sigma_2^2\nu''(B_2-,1)+(\mu_2-u)\nu'(B_2,2)-(\delta+\lambda_2)\nu(B_2,1)+\lambda_2\nu(B_2,2)+u\\
&&\leq0.
\end{eqnarray*}
In addition, for $x>B_2$, $\nu''(x,2)=0$, then we get
\begin{eqnarray*}
&&\frac{1}{2}\sigma_2^2\nu''(x,2)+(\mu_2-u)\nu'(x,2)-(\delta+\lambda_2)\nu(x,1)+\lambda_2\nu(x,2)+u\\
&&=(\mu_2-u)-(\delta+\lambda_2)\nu(x,1)+\lambda_2\nu(x,2)+u\\
&&=(\mu_2-u)-(\delta+\lambda_2)\nu(B_2,1)+\lambda_2\nu(B_2,2)+u+(\delta+\lambda_2)\big[\nu(B_2,1)-\nu(x,1)\big]+\lambda_2\big[\nu(x,2)-\nu(B_2,2)\big]\\
&&\leq(\delta+\lambda_2)\big[\nu(B_2,1)-\nu(x,1)\big]+\lambda_1\big[\nu(x,2)-\nu(B_2,2)\big]\\
&&=(\delta+\lambda_2)\big[\nu(b_1,1)+B_2-b_1-K-(\nu(b_1,1)+x-b_1-K)\big]\\
&&\quad +\lambda_2\big[\nu(b_2,2)+x-b_2-K-(\nu(b_2,2)+B_2-b_2-K)\big]\\
&&=(\delta+\lambda_2)(B_2-x)+\lambda_2(x-B_2)\\
&&=\delta(B_2-x)\leq0.
\end{eqnarray*}
Then we have
\[\max_{u\in[0,L]}\bigg\{\widetilde L_2(u)\nu(x,2)+u\bigg\}\leq0.\]
Combining with $\nu(x,2)=\mathcal{M}\nu(x,2)$, we know QVI is satisfied.

As far as the case in which $\nu'(0,1)=1$ and $\nu'(0,2)>1$, the verfication is similar, we omit it.
\end{proof}

\section{Conclusion}

In this paper, we stand the point of a manager to consider the optimal dividend problem. We assume the surplus of the company is affected by macroeconomic conditions. The manager have to make decision based on economic change which is described by a continuous-time Markov chain. The main contribution of this paper is that we break through the limitation of single dividend strategy to consider the mixed classical-impulse dividend strategy, and the value function as well as the optimal dividend strategy are explicitly derived.

%{\bf Acknowledgement:} The research of Xin Zhang is
% supported by the Discovery Grant from the Australian Research Council (ARC),
%(Project No.: DP1096243) and the National Natural Science Foundation of China(NSFC
%    grant No.11001139). The research of Hui Meng is
%supported by 121 Young Doctorial Development Fund Project for Central University of Finance and
%Economics(No.QBJJJ201004).

%\bibliography{C:/dividend}
%\bibliographystyle{plainnat}
\end{document}